\documentclass{article}
\usepackage{amsthm,amssymb}
\usepackage[margin=1in]{geometry}

\usepackage{graphicx} %
\usepackage{xcolor}
\usepackage[ruled,linesnumbered,noend]{algorithm2e}
\usepackage{eqparbox}
\usepackage{mathtools} 
\usepackage{enumitem}
\usepackage{thm-restate}
\usepackage[colorlinks, linkcolor=magenta, citecolor=magenta]{hyperref}

\usepackage{natbib}

\usepackage{complexity}

\usepackage{bm}
\usepackage{bbm}
\usepackage{xspace}
\usepackage{complexity}
\usepackage{amsthm}

\usepackage{physics2}
\usephysicsmodule{qtext.legacy,ab}

\newcommand{\delimit}[3]{
    \NewDocumentCommand#1{sm}{
    \IfBooleanTF##1
    {\left#2 ##2 \right#3}
    {\mathopen{#2} ##2 \mathclose{#3}}%
    }
}

\delimit\ceil\lceil\rceil
\delimit\floor\lfloor\rfloor
\delimit\ip\langle\rangle
\delimit\norm\lVert\rVert
\delimit\abs\lvert\rvert

\def\vb{{\bm{b}}}

\def\vg{{\bm{g}}}

\def\vm{{\bm{m}}}

\def\vq{{\bm{q}}}
\def\vr{{\bm{r}}}
\def\vs{{\bm{s}}}

\def\vu{{\bm{u}}}
\def\vv{{\bm{v}}}

\def\vx{{\bm{x}}}
\def\vy{{\bm{y}}}
\def\vz{{\bm{z}}}

\renewcommand\vec\bm

\newcommand{\mat}{\mathbf}

\def\mA{{\mathbf{A}}}

\def\cA{{\mathcal{A}}}

\def\cJ{{\mathcal{J}}}

\def\cQ{{\mathcal{Q}}}
\def\cR{{\mathcal{R}}}

\def\cU{{\mathcal{U}}}

\def\cX{{\mathcal{X}}}
\def\cY{{\mathcal{Y}}}

\DeclareMathOperator*{\argmax}{argmax}

\renewcommand{\R}{\mathbb R}

\let\ds\displaystyle
\let\op\operatorname
\let\eps\epsilon

\newcommand{\defeq}{:=}

\newcommand{\ie}{{\em i.e.}\xspace}
\newcommand{\eg}{{\em e.g.}\xspace}

\newcommand{\Reg}{\text{\rm \sffamily Reg}}

\usepackage{natbib}
\usepackage[capitalize,nameinlink]{cleveref}
\usepackage{autonum}

\newcommand{\diam}{D}
\newcommand{\mispred}{P}

\theoremstyle{plain}
\newtheorem{theorem}{Theorem}[section]
\newtheorem*{theorem*}{Theorem}

\newtheorem{lemma}[theorem]{Lemma}
\newtheorem{fact}[theorem]{Fact}
\newtheorem{proposition}[theorem]{Proposition}
\newtheorem{corollary}[theorem]{Corollary}

\theoremstyle{definition}
\newtheorem{definition}[theorem]{Definition}
\newtheorem{remark}[theorem]{Remark}
\newtheorem{assumption}[theorem]{Assumption}

\newcommand{\RM}{\texttt{RM}}
\newcommand{\PRM}{\texttt{PRM}}
\newcommand{\RMplus}{\texttt{RM}^+}
\newcommand{\PRMplus}{\texttt{PRM}^+}
\newcommand{\PCFRplus}{\texttt{PCFR}^+}
\newcommand{\IREGPRMplus}{\texttt{IREG-PRM}^+}
\newcommand{\IREGPRM}{\texttt{IREG-PRM}}
\newcommand{\IRPRM}{\texttt{IR-PRM}}
\newcommand{\IRPRMplus}{\texttt{IR-PRM}^+}
\newcommand{\IRCFRplus}{\texttt{IR-PCFR}^+}
\newcommand{\FTRL}{\texttt{FTRL}}
\newcommand{\MD}{\texttt{MD}}
\newcommand{\OGD}{\texttt{OGD}}
\newcommand{\AdOGD}{\texttt{AdOGD}}
\newcommand{\CFR}{\texttt{CFR}}
\newcommand{\DCFR}{\texttt{DCFR}}

\newcommand{\proj}{\Pi}
\newcommand{\tvx}{\tilde{\vx}}
\newcommand{\bvx}{\bar{\vx}}
\newcommand{\bvy}{\bar{\vy}}
\newcommand{\etao}{\eta}

\newcommand{\reg}{\mathsf{Reg}}
\newcommand{\utilbound}{B}
\let\cite\citep

\usepackage{tikz}
\usetikzlibrary{calc, arrows.meta, quotes, angles, intersections}

\usepackage{nicefrac}

\DontPrintSemicolon
\SetKwProg{Fn}{function}
\SetAlFnt{\relax}
\SetAlCapFnt{\relax}
\SetAlCapNameFnt{\relax}
\SetAlCapHSkip{0pt}

\SetKwInOut{Input}{Input}
\SetKwInOut{Output}{Output}
\SetKwFor{Repeat}{repeat}{}{}
\makeatletter
\newcommand{\Comment}[2][]{\hfill \eqparbox{#1@\thealgocf}{\color{gray} \itshape $\triangleright$ #2}}

\makeatother

\definecolor{briancolor}{rgb}{0, .5, 0}

\title{Scale-Invariant Regret Matching and Online Learning with Optimal Convergence: Bridging Theory and Practice in Zero-Sum Games}

\usepackage{authblk}
\author[1]{Brian Hu Zhang}
\author[2]{Ioannis Anagnostides}
\author[2,3]{Tuomas Sandholm}

\affil[1]{Massachusetts Institute of Technology}
\affil[2]{Carnegie Mellon University}
\affil[3]{Additional affiliations: Strategy Robot, Inc., Strategic Machine, Inc., Optimized Markets, Inc.}
\affil[ ]{\texttt{zhangbh}\texttt{@csail.mit.edu}, \texttt{\{ianagnos,sandholm\}}\texttt{@cs.cmu.edu}}

\begin{document}

\begin{titlepage}
\maketitle
\pagenumbering{gobble}

\begin{abstract}
    A considerable chasm has been looming for decades between theory and practice in zero-sum game solving through first-order methods. Although a convergence rate of $T^{-1}$ has long been established since Nemirovski's mirror-prox algorithm and Nesterov's excessive gap technique in the early 2000s, the most effective paradigm in practice is \emph{counterfactual regret minimization (CFR)}, which is based on \emph{regret matching} and its modern variants. In particular, the state of the art across most benchmarks is \emph{predictive} regret matching$^+$ (\texttt{PRM}$^+$), in conjunction with non-uniform averaging. Yet, such algorithms can exhibit slower $\Omega(T^{-1/2})$ convergence even in self-play.

    In this paper, we close the gap between theory and practice. We propose a new scale-invariant and parameter-free variant of \texttt{PRM}$^+$, which we call \texttt{IREG-PRM}$^+$. We show that it achieves $T^{-1/2}$ best-iterate and $T^{-1}$ (\ie, optimal) average-iterate convergence guarantees, while also being on par or even better relative to \texttt{PRM}$^+$ on benchmark games. From a technical standpoint, we draw an analogy between (\texttt{IREG}-)\texttt{PRM}$^+$ and optimistic gradient descent whose \emph{adaptive} learning rate depends on the misprediction error. The basic flaw of \texttt{PRM}$^+$ is that the ($\ell_2$-)norm of the regret vector---which can be thought of as the inverse of the learning rate---can decrease. By contrast, we design \texttt{IREG}-\texttt{PRM}$^+$ so as to maintain the invariance that the norm of the regret vector is nondecreasing. This enables us to derive an RVU-type bound for \texttt{IREG}-\texttt{PRM}$^+$, the first such property that does not rely on introducing additional hyperparameters to enforce smoothness. Reflecting this theoretical bridge, we find that the adaptive version of optimistic gradient descent we consider performs on par with \texttt{IREG}-\texttt{PRM}$^+$. This demystifies the effectiveness of the regret matching family \emph{vis-\`a-vis} more standard optimization techniques.

    Moreover, we extend our analysis beyond zero-sum games in normal form under a generalization of the classic Minty condition from variational inequalities that incorporates positive scaling weights. This weighted Minty condition is known to encompass, for example, (multi-player) harmonic games. We also observe that it captures CFR in zero-sum extensive-form games (with fully mixed Nash equilibria), establishing an intriguing new connection between CFR and harmonic games. Under the weighted Minty condition, we show that any algorithm satisfying a scale-invariant RVU property (such as \texttt{IREG}-\texttt{PRM}$^+$) has constant regret (in self-play) and $T^{-1/2}$ iterate convergence, thereby extending our convergence results to harmonic games and CFR in zero-sum extensive-form games. Unlike prior work in harmonic games, our algorithms do not require knowing the underlying weights by virtue of scale invariance.
\end{abstract}

\newpage
\tableofcontents

\end{titlepage}
\pagenumbering{arabic}

\section{Introduction}

\emph{Regret matching} ($\RM$) is a seminal online algorithm famously introduced by~\citet{Hart00:Simple}. $\RM$ keeps track of the cumulative \emph{regret} of each action so far and then proceeds by playing each action with probability proportional to its (nonnegative) regret. Its popularity can be attested by the many different variants that have been put forth over the years; most notably, \emph{regret matching$^+$} ($\RMplus$)~\citep{Tammelin14:Solving}, which truncates the negative coordinates of the regret vector to zero in each iteration; a generalization of both $\RMplus$ and $\RM$ called \emph{discounted regret matching} ($\texttt{DRM}$)~\citep{Brown19:Solving}, which discounts the cumulative regrets so as to alleviate the algorithm's inertia; and \emph{predictive regret matching}($^+$)~\citep{Farina21:Faster}, abbreviated as $\PRM$($^+$), which incorporates a prediction vector that intends to estimate the upcoming, future regret vector. All these algorithms converge---in a time-average sense---to the set of Nash equilibria in any zero-sum game when run in self-play~\citep{Freund99:Adaptive}.

The regret matching family is an indispensable component in state of the art algorithms for practical game solving in sequential decision problems, such as poker~\citep{Bowling15:Heads,Brown17:Superhuman,Brown19:Superhuman,Moravvcik17:DeepStack}, where one employs regret matching independently on each decision point---this is the \emph{counterfactual regret minimization} algorithm of~\citet{Zinkevich07:Regret}. Part of the appeal of $\RM$ and its variants in practice is that they are \emph{parameter free} and \emph{scale invariant}. Yet, their practical superiority has been bemusing from a theoretical standpoint. $\PRMplus$, the variant that typically performs best in practice---in conjunction with non-uniform averaging~\citep{Zhang24:Faster}---can converge at a rate of $\Omega(T^{-1/2})$~\citep{Farina23:Regret}, which is considerably slower \emph{vis-\`a-vis} other first-order algorithms that have a superior rate of $T^{-1}$; this includes the mirror-prox algorithm of~\citet{Nemirovski04:Prox}, the excessive gap technique of~\citet{Nesterov05:Excessive}, and the more recent \emph{optimistic} mirror descent algorithm~\citep{Rakhlin13:Online,Chiang12:Online}, which has the additional benefit of being compatible with the usual online learning framework.

Our goal in this paper is to close this chasm between theory and empirical performance, and, along the way, to demystify what makes the regret matching family so effective in practice. To put this into context, we should mention that \citet{Farina23:Regret}, who first identified the theoretical deficiency of $\PRMplus$, introduced a \emph{smooth} variant of regret matching that does attain the optimal $T^{-1}$ rate in zero-sum games. However, as noted by those authors, imposing smoothness comes at the cost of undermining practical performance. Indeed, practical experience suggests that part of what makes $\RM$ and its variants effective is precisely its \emph{lack} of smoothness, being much more aggressive than other algorithms such as (optimistic) gradient descent or multiplicative weights update. On top of that, the smooth variant necessitates tuning a certain hyper-parameter, which can be cumbersome in practice. Taking a step back, the crux is that existing techniques more broadly for establishing the optimal $T^{-1}$ rate in zero-sum games crucially hinge on additional hyperparameters to enforce smoothness, which has been at odds with practical performance.

\subsection{Our results}

We provide the first parameter-free and scale-invariant version of $\RM$ with a theoretically optimal $T^{-1}$ rate in zero-sum games. On top of that, it empirically performs on par or even better relative to $\PRMplus$ and other state of the art algorithms, as we demonstrate in~\Cref{sec:exp}. We thus bridge theory and practice in zero-sum game solving through first-order methods.

Our approach is driven by connecting (\texttt{P})$\RMplus$ to projected gradient descent \emph{with time-varying learning rate}. In particular, we think of the ($\ell_2$-)norm of the regret vector as serving as the inverse of the learning rate. From this perspective, $\PRMplus$ has a basic flaw: its ``learning rate'' can be increasing---that is, the norm of the regret vector can be decreasing. This fact was already noted by~\citet{Farina23:Regret}, illustrated in~\Cref{fig:decreasingregs}, middle. It is based on the zero-sum game with payoff matrix
\begin{equation}\label{eq:counterexample}
    \mat{A} = \begin{pmatrix}
        3 & 0 & -3 \\
        0 & 3 & -4 \\
        0 & 0 & 1
    \end{pmatrix}.
\end{equation}
Incidentally, this is also a game where, numerically, $\PRMplus$ has a slow convergence rate of $\Omega(T^{-1/2})$. While a player having small---indeed, negative (\Cref{fig:decreasingregs}, middle)---regret is not a problem \emph{per se}, it results in destabilizing the iterates of that player, which in turn makes it \emph{harder for its opponent} to predict the next utility.

\begin{figure}[!ht]
    \centering
    \includegraphics[width=0.32\textwidth]{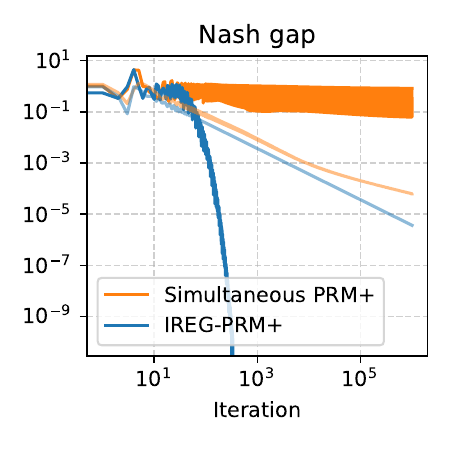}
    \includegraphics[width=0.32\textwidth]{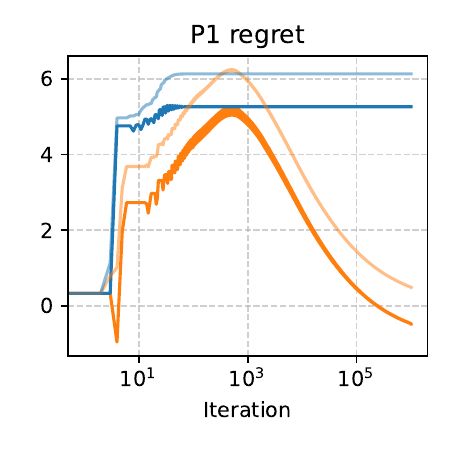}
    \includegraphics[width=0.32\textwidth]{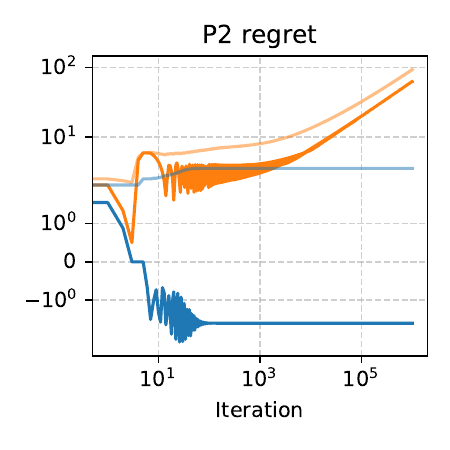}
    \caption{$\IREGPRMplus$ and simultaneous $\PRMplus$ on the counterexample game \eqref{eq:counterexample}. In the left plot, the dark lines and light lines show the Nash gap of the last iterate and average iterate, respectively. In the middle and right plots, the dark lines show the actual regret, and the light lines show the $\ell_2$ norm of the regret vector.}
    \label{fig:decreasingregs}
\end{figure}

The variant that we propose, coined \emph{increasing regret extra-gradient predictive regret matching$^+$}, or $\IREGPRMplus$ for short, maintains the basic invariance that the regret vector is nondecreasing (\Cref{fig:decreasingregs}, light blue lines on the middle and right plots). It does so through a judicious shift in the predicted regret vector, computed by solving a certain one-dimensional optimization problem; we show that this can be done exactly in linear time (\Cref{sec:gamma}), so the per-iteration complexity of $\IREGPRMplus$ is on par with $\RM$ and its variants. Furthermore, as the name suggests, $\IREGPRMplus$ also makes use of an extra-gradient step to come up with the next prediction in each step. It should be noted that $\IREGPRMplus$ is an instantiation of a more general family that we introduce, namely $\IRPRMplus$. $\IRPRMplus$ is parameterized by a sequence of predictions, and is compatible with the usual online learning framework.

From a technical standpoint, the key fact about $\IREGPRMplus$ is that it satisfies a certain \emph{RVU bound} (per~\Cref{def:RVU}). This property was crystallized by~\citet{Syrgkanis15:Fast} and has been at the heart of designing faster no-regret dynamics in games. While algorithms such as optimistic $\FTRL$ and optimistic $\MD$ have this property, we establish that $\IREGPRMplus$ is the first parameter-free, scale-invariant algorithm that admits a certain RVU-type bound (\Cref{th:rvu}). In turn, this suffices to show that $\IREGPRMplus$ has both the coveted $T^{-1}$ average-iterate---which is optimal among algorithms performing uniform averaging~\citep{Daskalakis15:Near}---as well as $T^{-1/2}$ (best-)iterate convergence (\Cref{cor:iregpcfr}), making it the first parameter-free, scale-invariant algorithm with this property; among other reasons, this is important because the last iterate often converges significantly faster than the average, as we demonstrate in~\Cref{sec:exp}.

Our second, more conceptual contribution is to bridge the regret matching family with more traditional gradient-based algorithms in optimization. Specifically, our analysis (\Cref{sec:gdasrmplus}) reveals a tight connection between $\IREGPRMplus$ and an adaptive version of optimistic gradient descent that we consider ($\AdOGD$). The key idea behind $\AdOGD$ is a learning rate sequence that adapts based on the misprediction error. We find that $\AdOGD$ enjoys an RVU-type bound similar to the one we obtain for $\IREGPRMplus$ (\Cref{theorem:RVU-adOGD}), which again leads to the optimal $T^{-1}$ rate for the average strategies (\Cref{cor:opt-AdOGD}) together with $T^{-1/2}$ iterate convergence (\Cref{cor:lastiterate}). What is more, our experiments reveal that $\AdOGD$ performs, for the most part, on par with $\IREGPRMplus$. To our knowledge, $\AdOGD$ is the first gradient descent-type algorithm that closely matches the state of the art in zero-sum extensive-form games. From a conceptual standpoint, this demystifies the effectiveness of $\RM$ and its variants relative to more traditional approaches in optimization.

Finally, we extend our analysis beyond zero-sum games in normal form under a generalization of the classic Minty condition from variational inequalities---itself an extension of monotonicity---that incorporates positive scaling weights; we call this the \emph{weighted Minty condition} (\Cref{assumption:first}). It is well-known that this condition encompasses (multi-player) \emph{harmonic games}~\citep{Candogan11:Flows,Abdou22:Decomposition,Legacci24:No}. Surprisingly, we find that it also captures CFR in two-player extensive-form games (with fully mixed Nash equilibria), making an intriguing new connection between CFR and harmonic games. We show that any algorithm satisfying our scale-invariant RVU property (\Cref{assumption:harmonic-rvu})---such as $\AdOGD$ and $\IREGPRMplus$---enjoys $T^{-1/2}$ iterate convergence and constant regret (\Cref{th:fast-convergence,thm:fast-bestiterate,th:efg}). This is the first iterate convergence guarantee in zero-sum extensive-form games obtained through a scale-invariant CFR-based algorithm (\emph{cf.}~\citealp{Liu23:Power,Meng25:Efficient,Lee21:Last}). Furthermore, compared to previous work on harmonic games by~\citet{Legacci24:No}, our algorithms do not require knowing the underlying weights by virtue of scale invariance.

\subsection{Further related work}
\label{sec:related}

The effectiveness of regret matching as a practical zero-sum game solving algorithm was first recognized by~\citet{Zinkevich07:Regret}, who introduced the counterfactual regret minimization ($\CFR$) algorithm for (imperfect-information) extensive-form games. $\CFR$ can be thought of as a framework that prescribes using a separate regret minimizer in each decision point of the tree; it is sound no matter what no-regret algorithms are employed~\citep{Farina19:Regret}, but by far the most effective approach in practice has been through the regret matching family. Following the paper of~\citet{Hart00:Simple} that introduced regret matching, many different variants and extensions have been proposed to speed up its performance~\citep{Xu24:Minimizing,Cai25:Last,Chakrabarti24:Extensive,Meng25:Asynchronous,Farina21:Faster,Tammelin14:Solving,Brown19:Solving,Marden07:Regret,Hart03:Regret}. $\PRMplus$, introduced by~\citet{Farina21:Faster}, is the state of the art algorithm across most benchmarks, and its performance can be further boosted by employing a non-uniform averaging scheme~\citep{Zhang24:Faster}. An interesting connection made by~\citet{Farina21:Faster} links $\RM$ to $\FTRL$ and $\RMplus$ to $\MD$ through the lens of Blackwell approachability~\citep{Blackwell56:analog}. However, as was mentioned earlier, $\PRMplus$ can suffer from slow convergence rate of $\Omega(T^{-1/2})$, and this is so even in $3 \times 3$ normal-form zero-sum games~\citep{Farina23:Regret}. This perhaps partly explains why $\PRMplus$ is inferior than other algorithms in some benchmark games---namely, ones based on poker~\citep{Farina21:Faster}.

At the same time, we have seen that first-order methods with a superior $T^{-1}$ rate have been known before $\CFR$ came to the fore. While they have shown some promise in solving large zero-sum extensive-form games~\citep{Hoda10:Smoothing,Kroer18:Solving,Farina21:Better}, they are lagging behind $\RM$ and its variants when it comes to larger games. Finally, in relation to the $\AdOGD$ algorithm that we introduce, we stress that many adaptive algorithms have been proposed and analyzed in the context of zero-sum games and beyond (\emph{e.g.}, \citealp{Antonakopoulos21:Adaptive,Antonakopoulos19:Adaptive,Alacaoglu20:New,Tsuchiya25:Corrupted}), but their practical performance in extensive-form games has remained unexplored; we fill this gap by benchmarking $\AdOGD$ across several games.

The weighted Minty condition we consider (\Cref{assumption:first}) is a generalization of the classic Minty condition from variational inequalities~\citep{Minty67:Generalization}. The weighted Minty condition encompasses harmonic games~\citep{Candogan11:Flows},  even under the more general definition of~\citet{Legacci24:No}. A polynomial-time algorithm for solving variational inequalities under the weighted Minty condition was provided by~\citet{Anagnostides25:Polynomial}. When specialized to games, the Minty condition is equivalent to imposing that the sum of the players' regrets is nonnegative. The weighted version of that condition  was recently treated by~\citet{Zhou25:Pointwise}.

Concurrent work develops new parameter-free versions of regret matching with strong practical performance~\citep{Anonymous26:Faster}. However, unlike our algorithm, theirs is not scale invariant.

\section{Background}

Before we proceed, we introduce some basic background on regret minimization in the context of (two-player) zero-sum games. Our main focus in this paper lies primarily in solving the bilinear saddle-point problem
\begin{equation}
    \label{eq:bspp}
    \max_{\vx \in \cX} \min_{\vy \in \cY} \vx^\top \mat{A} \vy,
\end{equation}
where $\cX$ and $\cY$ are convex and compact subsets of a Euclidean space. A canonical case arises when $\cX$ and $\cY$ are probability simplices, in which case~\eqref{eq:bspp} is known to be equivalent to linear programming (\emph{e.g.},~\citealp{Stengel24:Zero}). In what follows, we refer to the bilinear saddle-point problem~\eqref{eq:bspp} as a zero-sum game between Player $\cX$ and Player $\cY$.

The most effective approach to solving zero-sum games in practice is through iterative first-order algorithms, and particularly the framework of \emph{regret minimization}. The key premise here is that the two players repeatedly play the game for multiple rounds $t = 1, \dots, T$. At the beginning of each round $t \in [T]$, the players specify their strategies, $\vx^{(t)} \in \cX$ and $\vy^{(t)} \in \cY$. Then they observe as utility feedback the matrix-vector products $\vu^{(t)}_\cX \defeq \mat{A} \vy^{(t)}$ and $\vu^{(t)}_\cY \defeq - \mat{A}^\top \vx^{(t)}$, respectively; this is the usual simultaneous update setup, but in the sequel we also consider \emph{alternating} updates. \Cref{alg:learningsetups} (\Cref{sec:learningsetup}) contains a detailed overview of all these learning setups.

The \emph{regret} of Player $\cX$ is defined as
\begin{equation}
    \label{eq:regret}
    \reg^{(T)}_\cX \defeq \max_{\vx^* \in \cX} \sum_{t=1}^T \langle \vx^* - \vx^{(t)}, \vu_\cX^{(t)} \rangle,
\end{equation}
and similarly for Player $\cY$; in~\eqref{eq:regret}, $\langle \cdot, \cdot \rangle$ denotes the inner product.

A key connection between online learning and game theory is that players whose regret grows sublinearly with the time horizon $T$ converge, in a \emph{time-average} sense, to minimax equilibria~\citep{Freund99:Adaptive}. Specifically, in non-asymptotic terms, we measure distance to optimality of a point $(\vx, \vy) \in \cX \times \cY$ through the \emph{duality gap},
\begin{equation}
    \label{eq:dualitygap}
    (\vx, \vy) \mapsto \max_{\vx^* \in \cX} \langle \vx^*, \mat{A} \vy \rangle - \min_{\vy^* \in \cY} \langle \vy^*, \mat{A}^\top \vx \rangle.
\end{equation}

\begin{proposition}
    \label{prop:folklore}
    Let $\bvx^{(T)} \defeq \frac{1}{T} \sum_{t=1}^T \vx^{(t)}$ and $\bvy^{(T)} \defeq \frac{1}{T} \sum_{t=1}^T \vy^{(t)}$. If the players have regret $\reg_\cX^{(T)}$ and $\reg_\cY^{(T)}$ after $T$ repetitions of a zero-sum game, respectively, the average strategy profile $(\bvx^{(T)}, \bvy^{(T)})$ has duality gap equal to $\frac{1}{T} \left( \reg_\cX^{(T)} + \reg_\cY^{(T)} \right)$.
\end{proposition}
That is, the convergence of the average strategies is driven by the \emph{sum} of the players' regrets. We will also use the following basic fact.

\begin{fact}
    \label{fact:obvious}
    In any zero-sum game, $\reg_\cX^{(T)} + \reg_\cY^{(T)} \geq 0$.
\end{fact}
This holds simply because the sum of the regrets is equal to the duality gap of the average strategies~\eqref{eq:dualitygap}, which is in turn nonnegative.

A powerful technique for bounding the sum of the players' regrets in a game is the \emph{RVU property} crystallized by~\citet{Syrgkanis15:Fast}, which stands for ``regret bounded by variation in utilities.''

\begin{definition}[RVU bound; \citealp{Syrgkanis15:Fast}]
    \label{def:RVU}
    A regret minimization algorithm that produces a sequence of strategies $(\vx^{(t)})_{t=1}^T$ under a sequence of utilities $(\vu^{(t)})_{t=1}^T$ satisfies the \emph{RVU} bound with respect to $(\alpha, \beta, \gamma) \in \R^3_{> 0}$ and a primal-dual norm pair $(\|\cdot\|, \|\cdot\|_*)$ if
    \begin{equation}
        \reg^{(T)} \leq \alpha + \beta \sum_{t=2}^T \| \vu^{(t)} - \vu^{(t-1)} \|_*^2 - \gamma \sum_{t=2}^T \|\vx^{(t)} - \vx^{(t-1)} \|^2.\label{eq:rvu}
    \end{equation}
\end{definition}
This property is satisfied for both optimistic mirror descent and optimistic follow the regularized leader with $\alpha \propto 1/\eta$, $\beta = \eta$, and $\gamma \propto 1/\eta$, where $\eta$ is the learning rate~\citep{Syrgkanis15:Fast}. This in turn implies that, if all players use those algorithms to update their strategies, the sum of their regrets will remain bounded~\citep{Syrgkanis15:Fast}.  

A key ingredient that has been used to obtain fast convergence is the smoothness (or stability) of the iterates: $\|\vx^{(t)} - \vx^{(t-1)} \| \leq O(\eta)$.\footnote{A notable recent exception is optimistic fictitious play: \citet{Lazarsfeld25:Optimism} showed that it has constant regret, but only for $2 \times 2$ games.} Unfortunately, this property does not hold for the regret matching family~\citep{Farina23:Regret}, which has been the main obstacle in overcoming the $T^{-1/2}$ barrier in the rate of convergence.

\paragraph{Regret matching} Regret matching ($\RM$)~\cite{Hart00:Simple} and regret matching$^+$ ($\RMplus$)~\citet{Tammelin14:Solving} are standard algorithms for regret minimization over the simplex $\cX = \Delta(n)$. Their pseudocode is given in~\Cref{alg:rmplus}. The only difference between $\RM$ and $\RMplus$ is that the latter truncates the regret vector to zero in each iteration. 
This seemingly innocuous modification translates to much faster convergence in practice~\citep{Tammelin15:Solving}.

\begin{algorithm}[h]
\caption{$\RM$ and $\RMplus$}
\label{alg:rmplus}
\Fn{{\sc Initialize}()}{
$\vr^{(1)} \gets \vec 0$
}
\Fn{{\sc NextStrategy}()}{
{\bf if} {$[\vr^{(t)}]_+ = \vec 0$} {\bf then} {\Return {\em an arbitrary strategy}
} \Comment{$[\cdot]_+$ is the element-wise positive part} \\
\Return $\vx^{(t)} \gets [\vr^{(t)}]_+ / \norm{[\vr^{(t)}]_+}_1$ 
}
\Fn{{\sc ObserveUtility}(utility $\vu^{(t)} \in \R^n$)}{
$\vg^{(t)} \gets \vu^{(t)} - \ip{\vu^{(t)}, \vx^{(t)}}$  \\
$\vr^{(t+1)} \gets [\vr^{(t)} + \vg^{(t)}]_+$ \Comment{$\vr^{(t+1)} \gets \vr^{(t)} + \vg^{(t)}$ for $\RM$}
}  
\end{algorithm}

\paragraph{Scale invariance.} \citet{Chakrabarti24:Extensive} hypothesized that {\em scale invariance} is a crucial property that enables the fast performance of regret minimizers in solving zero-sum games. Regret matching variants generally satisfy this property, but algorithms with guaranteed $O(1/T)$ convergence, such as $\OGD$, do not. Formally, a no-regret learning algorithm is {\em scale-invariant} if, for every utility sequence $\vu^{(1)}, \dots, \vu^{(T)}$ and every constant $c > 0$, the algorithm produces the same sequence of strategies $\vx^{(1)}, \dots, \vx^{(T+1)}$ on utility sequences $\vu^{(1)}, \dots, \vu^{(T)}$ and $c\vu^{(1)}, \dots, c\vu^{(T)}$. 

Before our paper, there were no scale-invariant regret minimization algorithms with proven best-iterate or $O(1/T)$ average-iterate convergence rates. Indeed, variants of regret matching that have such provable guarantees achieve them while losing the scale-invariance property (\eg,~\citealp{Farina23:Regret}). On the other hand, the fastest practical algorithms for learning in games, which are all based on $\RM$, are scale invariant and have $O(1/T)$ or even better convergence in practice, but lack theoretical guarantees. Our paper's main goal is to close this theory-practice gap.

\section{A near-optimal variant of regret matching}
\label{sec:RM}

In this section, we develop variants of regret matching, $\IRPRM$ and $\IRPRMplus$ that satisfy an RVU-type bound, and therefore lead to fast convergence guarantees. Motivated by the counterexample in \Cref{fig:decreasingregs}, the main intuition behind our algorithm is that it maintains predictivity while also enforcing the constraint that the $\ell_2$ norm never decreases. The result is \Cref{alg:prmp-mod}. As a warm-up, the reader can first go through the more standard analysis of adaptive optimistic gradient descent ($\AdOGD$) in~\Cref{sec:AdOGD}, as it conveys some key ideas that also appear in the analysis of $\IRPRMplus$.

\begin{algorithm}[h]
\caption{$\IRPRM$ and $\IRPRMplus$}
\label{alg:prmp-mod}
\Fn{{\sc Initialize}()}{
$\tilde\vr^{(1)} \gets{}$arbitrary vector in $\R^n_{\ge 0}$\\
$\tilde\vx^{(1)} \gets \tilde\vr^{(1)} / \norm{\tilde\vr^{(1)}}_1$  \Comment{if $\tilde\vr^{(1)} = \vec 0$, set an arbitrary strategy}
}
\Fn{{\sc NextStrategy}(prediction $\vm^{(t)} \in \R^n$)}{
\If{$[\tilde\vr^{(t)}]_+ = \vec 0$}{
$\vm^{(t)} \gets \vec 0$\\
\Return $\vx^{(t)} \gets \tilde\vx^{(t)}$
}
let $\gamma \in \R$ be s.t. $\norm{[\tilde\vr^{(t)} + \vm^{(t)} - \gamma \vec 1]_+}_2 = \norm{\tilde\vr^{(t)}}_2$ \\
$\vr^{(t)} \gets \tilde\vr^{(t)} + \vm^{(t)} - \gamma \vec 1$ \\
\Return $\vx^{(t)} \gets [\vr^{(t)}]_+ / \norm{[\vr^{(t)}]_+}_1$ 
}
\Fn{{\sc ObserveUtility}(utility $\vu^{(t)} \in \R^n$)}{
$\vg^{(t)} \gets \vu^{(t)} - \vm^{(t)} - \ip{\vu^{(t)} - \vm^{(t)}, \vx^{(t)}}$  \\
$\tilde\vr^{(t+1)} \gets [\vr^{(t)} + \vg^{(t)}]_+$ \Comment{$\tilde\vr^{(t+1)} \gets \vr^{(t)} + \vg^{(t)}$ for $\IRPRM$}\\
$\tilde\vx^{(t+1)} \gets \tilde\vr^{(t)} / \norm{\tilde\vr^{(t)}}_1$  \Comment{if $\tilde\vr^{(t)} = \vec 0$, set $\tilde\vx^{(t+1)} \gets \vx^{(t)}$} \\
}
\end{algorithm}
We now give some intuition for the algorithm. Consider the standard $\RM^{(+)}$ algorithm (equivalent to \Cref{alg:prmp-mod} in the case $\vm^{(t)} := \vec 0$). Without predictions, these satisfy the nondecreasing regret norm condition:
\begin{lemma}\label{lem:nondec-regret}
    For $\RM^{(+)}$, $\norm{[\vr^{(t+1)}]_+}_2 \ge \norm{[\vr^{(t)}]_+}_2$.
\end{lemma}
\begin{proof}
    Since $\vx^{(t)} \propto [\vr^{(t)}]_+$, we have $\ip{\vg^{(t)}, [\vr^{(t)}]_+} = 0$. Thus, 
    \begin{align}
        \norm{[\vr^{(t)}]_+}_2^2 = \ip{[\vr^{(t)}]_+ + \vg^{(t)}, [\vr^{(t)}]_+} = \ip{\vr^{(t+1)} + [\vr^{(t)}]_-, [\vr^{(t)}]_+} \le \ip{[\vr^{(t+1)}]_+, [\vr^{(t)}]_+}
    \end{align}
    which is only possible if $\norm{[\vr^{(t+1)}]_+}_2 \ge \norm{[\vr^{(t)}]_+}_2$.
\end{proof}

We can think of $\IRPRM^{(+)}$ by using $\RM^{(+)}$ as a ``black-box subroutine''. Notice the following equivalence: $\IRPRM^{(+)}$ accepting a prediction $\vm^{(t)}$ and then a utility $\vu^{(t)}$ has the same effect as $\RM^{(+)}$ accepting the utility $\vm^{(t)}/K$  (without any prediction) repeatedly $K$ times (in the limit $K \to \infty$), then outputting the strategy $\vx^{(t)}$, then accepting the utility $\vu^{(t)} - \vm^{(t)}$ in a single step. To see the equivalence, notice that after accepting $\vm^{(t)}$ in infinitesimally small increments, the resulting regret vector $\vr^{(t)}$ must have the form $\tilde\vr^{(t)} + \vm^{(t)} - \gamma \vec 1$ for some $\gamma$, and $\norm{[\vr^{(t)}]_+}_2 = \norm{\tilde\vr^{(t)}}_2$ since $[\tilde\vr^{(t)}]_+$ can only ever move perpendicular to itself, and therefore cannot change in norm. Therefore, $\IRPRM^{(+)}$ essentially {\em implements} this ``infinitesimal prediction'' version of $\RM^{(+)}$, and hence inherits the convenient properties of $\RM^{(+)}$, namely, its regret bound and nondecreasing regret vector norm guarantee. The prediction step of $\IRPRM^{(+)}$ is illustrated in \Cref{fig:irprmplus}.

\begin{figure}
    \centering
\scalebox{0.65}{\begin{tikzpicture}[
    scale=1.5,
    >=Latex, %
    font=\small
]

    \def\R{4}           %
    \def\AngOne{65}     %
    \def\AngTwo{19}     %
    
    \coordinate (O) at (0,0);
    \coordinate (P1) at (\AngOne:\R); %
    \coordinate (P2) at (\AngTwo:\R); %
    
    \draw[thin, black!70] (-0.5,0) -- (6,0); %
    \draw[thin, black!70] (0,-0.5) -- (0,5); %

    \draw[gray, thick] (95:\R) arc (95:-10:\R);

    \coordinate (FarPoint) at (\AngTwo:7);

    \coordinate (H) at ($(O)!(P1)!(P2)$);
    
    \coordinate (H_ext) at ($(P1)!1.2!(H)$);

    \coordinate (H_extext) at ($(P1)!1.4!(H)$);
    
    \draw[gray, thick] (P1) -- (H_extext);

    \begin{scope}[rotate=\AngTwo]
        \draw[fill=gray!20, draw=black!50] ($(O)!1!(H)$) -- ++(0, 0.25) -- ++(-0.25, 0) -- ++(0, -0.25);
    \end{scope}

    \draw[->, red!90!white, thick, line width=1.2pt] (O) -- (P1) node[pos=1.05, anchor=south west,xshift=-10pt,yshift=-5pt] {$\tilde{\vec{r}}^{(t)}$};
    \draw[->, red!90!white, thick, line width=1.2pt] (O) -- (P2) node[anchor=north west, xshift=2pt,yshift=10pt] {$\vec{r}^{(t)}$};

    \coordinate (BlueTip) at ($(P1) + (4.5, 0)$);
    \draw[->, blue!80!white, very thick] (P1) -- (BlueTip) node[midway, above, text=blue] {$\vec{m}^{(t)}$};

    \draw[gray, dotted, thick] (H_ext) -- (BlueTip);

    \draw[green!40!black, line width=1.5pt, line cap=round] (H_ext) -- (P2);

    \fill[black] (P1) circle (0.8pt);
    \fill[black] (P2) circle (0.8pt);

\end{tikzpicture}}
    \caption{An illustration of the prediction step of $\IRPRM^{(+)}$. After accepting the prediction $\vm^{(t)}$, the regret vector $\vr^{(t)}$ must lie somewhere on the dotted line. Since $\vr^{(t)}$ and $\tilde\vr^{(t)}$ are defined to have the same $\ell_2$ norm,  $\vr^{(t)}$ is fixed. The {\em true} regret vector follows the gray line perpendicular to $\vr^{(t)}$ passing through $\tilde\vr^{(t)}$; therefore, the thick green line segment marks the amount by which the ``naive'' regret analysis overestimates the true regret; lower-bounding the length of the green segment therefore yields the negative term of the RVU bound.}
    \label{fig:irprmplus}
\end{figure}

\Cref{alg:prmp-mod} is only scale-invariant with the setting $\vr^{(1)} = \vec 0$. Indeed, this is the setting for which we will show convergence to equilibrium, and which we will use in experiments. However, we state the more general version with any regret vector initialization because it is useful for analysis.

In \Cref{sec:gamma} we give an $O(n)$-time algorithm for computing the value $\gamma$ required by \Cref{alg:prmp-mod}. Thus, every iteration takes linear time.

\subsection{An RVU bound for $\IRPRM^{(+)}$}

We now show an RVU-type bound for \Cref{alg:prmp-mod}. Intuitively, the bound follows by the following argument: accepting the utility $\vm^{(t)}$ in infinitesimally small increments leads to a regret vector $\vr^{(t)}$, but $\vr^{(t)}$ actually {\em overestimates} the true regret, because the true regret is what was incurred by playing $\vx^{(t)}$ against $\vm^{(t)}$, whereas the algorithm moved from $\tilde\vx^{(t)}$ to $\vx^{(t)}$ continuously, playing some strategy in between. Lower-bounding the size of the overestimate will lead to the RVU bound.

\begin{theorem}[RVU bound for $\IRPRM^{(+)}$]\label{th:rvu}
    The regret of $\IRPRM$ and $\IRPRMplus$ is bounded by
    \begin{align}
        \sqrt{\norm{\tilde\vr^{(1)}}_2^2 + \sum_{t=1}^{T} \norm{\vg^{(t)}}_2^2} - \frac{1}{2n} \sum_{t=1}^{T}  \norm{[\tilde\vr^{(t)}]_+}_2 \norm{\vx^{(t)} - \tilde\vx^{(t)}}_2^2 
    \end{align}
\end{theorem}
\begin{proof} For notation, let $\tilde\vr^{(t+1)}_*$ be the true regret vector after $t$ timesteps, and let $\vr^{(t)}_*$ be what $\tilde\vr^{(t+1)}$ would have been if $\vu^{(t)} = \vm^{(t)}$. That is, they are defined by the recurrences
\begin{align}
    \tilde\vr^1_* = \vec 0, \qq{} \tilde\vr^{(t+1)}_* = \tilde\vr^{(t)}_* + \vu^{(t)} - \ip{\vu^{(t)}, \vx^{(t)}}, \qq{and} \vr^{(t)}_* = \tilde\vr^{(t)}_* + \vm^{(t)} - \ip{\vm^{(t)}, \vx^{(t)}}.
\end{align}
    We will first element-wise lower-bound the vector 
    \begin{align}
        \tilde\vr^{(T+1)} - \tilde\vr^{(T+1)}_* = \sum_{t=1}^{T} \ab[(\tilde\vr^{(t+1)} - \vr^{(t)}) - (\tilde\vr^{(t+1)}_* - \vr^{(t)}_*) + (\vr^{(t)} - \tilde\vr^{(t)}) - (\vr^{(t)}_* - \tilde\vr^{(t)}_*) ],
    \end{align}
    \ie, the amount by which $\tilde\vr^{(T+1)}$ overestimates the true regret vector.    We have
    $\tilde\vr^{(t+1)} \ge \vr^{(t)} + \vg^{(t)}$ by construction of the algorithm and $\tilde\vr^{(t+1)}_* = \vr^{(t)}_* + \vg^{(t)}$ by definition. Subtracting these gives $(\tilde\vr^{(t+1)} - \vr^{(t)}) - (\tilde\vr^{(t+1)}_* - \vr^{(t)}_*) \ge \vec 0$. It thus suffices to bound $(\vr^{(t)} - \tilde\vr^{(t)}) - (\vr^{(t)}_* - \tilde\vr^{(t)}_*)$. We claim that
    $$(\vr^{(t)} - \tilde\vr^{(t)}) - (\vr^{(t)}_* - \tilde\vr^{(t)}_*) \ge \frac{1}{2n} \norm{[\tilde\vr^{(t)}]_+}_2 \norm{\vx^{(t)} - \tilde\vx^{(t)}}_2^2$$
    (element-wise).     This would complete the proof, because then from the usual analysis of $\RM$, we have \begin{align}\norm{[\tilde\vr^{(T+1)}]_+}_2^2 \le \norm{\tilde\vr^{(1)}}_2^2 + \sum_{t=1}^{T} \norm{\vg^{(t)}}_2^2, \label{eq:rm_regret_bound}
    \end{align}
    and therefore
    \begin{align}
        \tilde\vr^{(T+1)}_*  &\le \tilde\vr^{(T+1)}  - \frac{1}{2n} \sum_{t=1}^{T}  \norm{[\tilde\vr^{(t)}]_+}_2 \norm{\vx^{(t)} - \tilde\vx^{(t)}}_2^2.
        \\&\le \norm{[\tilde\vr^{(T+1)}]_+}_2 - \frac{1}{2n} \sum_{t=1}^{T}  \norm{[\tilde\vr^{(t)}]_+}_2 \norm{\vx^{(t)} - \tilde\vx^{(t)}}_2^2
        \\&\le \sqrt{\norm{\tilde\vr^{(1)}}_2^2 + \sum_{t=1}^{T} \norm{\vg^{(t)}}_2^2} - \frac{1}{2n} \sum_{t=1}^{T}  \norm{[\tilde\vr^{(t)}]_+}_2 \norm{\vx^{(t)} - \tilde\vx^{(t)}}_2^2.
    \end{align}
    We now prove the claim. We will drop the time indices for notational simplicity. If $\tilde\vr \le \vec 0$, the claim is trivial: the right-hand side is $0$ by definition, and the left-hand side is zero since $\vm$ is defined to be $\vec 0$ in this case.
    Otherwise, by definition, we have $\ip{\vr_* - \tilde\vr_*, \vx} = 0$. Since $\vx \propto [\vr]_+$, this also implies  $\ip{\vr_* - \tilde\vr_*, [\vr]_+} = 0$. Moreover, we have
    \begin{align}
    \ip{\vr - \tilde\vr, [\vr]_+}  &= \ip{[\vr]_+ - \tilde\vr, [\vr]_+} 
    \\&= \norm{[\vr]_+}_2^2 - \ip{\tilde\vr, [\vr]_+} 
    \\&= \frac12 \norm{[\vr]_+}_2^2 + \frac12 \norm{\tilde\vr}_2^2 - \ip{\tilde\vr, [\vr]_+} 
    \\&\ge \frac12 \norm{[\vr]_+}_2^2 + \frac12 \norm{[\tilde\vr]_+}_2^2 - \ip{[\tilde\vr]_+, [\vr]_+} 
    \\&= \frac12 \norm{[\vr]_+ - [\tilde\vr]_+}_2^2 
    \end{align}
    where the third equality follows from the fact that $\gamma$ is chosen so that $\norm{[\vr]_+}_2 = \norm{[\tilde\vr]_+}_2$. But we also have $
        (\vr - \tilde\vr) - (\vr_* - \tilde\vr_*) = (\ip{\vm, \vx} - \gamma) \vec 1. 
    $
    Thus, in particular, we have $\ip{\vm, \vx} - \gamma \ge 0$ and
    \begin{align}
        (\ip{\vm, \vx} - \gamma) \cdot \norm{[\vr]_+}_1 &= \norm{(\vr - \tilde\vr) - (\vr_* - \tilde\vr_*)}_\infty \cdot \norm{[\vr]_+}_1 
        \\&\ge \ip{(\vr - \tilde\vr) - (\vr_* - \tilde\vr_*), [\vr]_+}
        \\&\ge \frac12 \norm{[\vr]_+ - [\tilde\vr]_+}_2^2 
        \\&\ge \frac{1}{2n} \norm{[\tilde\vr]_+}_2^2\cdot   \norm{\vx - \tilde\vx}_2^2 
    \end{align}
    where in the last line we use the fact that the map $\vz \mapsto \vz/\norm{\vz}_1$ is $\sqrt{n}$-Lipschitz in $\ell_2$ norm on the unit $\ell_2$-ball $\norm{\vz}_2 = 1$. Since $\norm{\cdot}_1 \ge \norm{\cdot}_2$, we conclude 
    \begin{align}\ip{\vm, \vx} - \gamma \ge \frac{1}{2n} \norm{[\tilde\vr]_+}_2\cdot   \norm{\vx - \tilde\vx}_2^2. \tag*\qedhere\end{align}
\end{proof}

In particular, if $0 \ne \norm{\tilde\vr^{(1)}}_2 =: 1/\eta$, then, using the fact that the (nonnegative parts of the) regret vectors have nondecreasing $\ell_2$ norm, we get
\begin{align}
    & \sqrt{\frac{1}{\eta^2} + \sum_{t=1}^{T} \norm{\vg^{(t)}}_2^2} - \frac{1}{2n} \sum_{t=1}^{T}  \norm{[\vr^{(t)}]_+}_2 \norm{\vx^{(t)} - \tilde\vx^{(t)}}_2^2
    \\&= \frac{1/\eta^2 + \sum_{t=1}^{T} \norm{\vg^{(t)}}_2^2}{\sqrt{1/\eta^2 + \sum_{t=1}^{T} \norm{\vg^{(t)}}_2^2}} - \frac{1}{2n} \sum_{t=1}^{T}  \norm{[\vr^{(t)}]_+}_2 \norm{\vx^{(t)} - \tilde\vx^{(t)}}_2^2
    \\&\le \frac{1/\eta^2 + \sum_{t=1}^{T} \norm{\vg^{(t)}}_2^2}{1/\eta} - \frac{1}{2n} \sum_{t=1}^{T}  \frac{1}{\eta} \norm{\vx^{(t)} - \tilde\vx^{(t)}}_2^2
    \\&= \frac{1}{\eta} + \eta \sum_{t=1}^{T} \norm{\vg^{(t)}}_2^2 - \frac{1}{2n\eta} \sum_{t=1}^{T}  \norm{\vx^{(t)} - \tilde\vx^{(t)}}_2^2
\end{align}
which is more similar to the standard RVU bound (\Cref{def:RVU}).

\section{Discussion: relationship between GD and RM+}
\label{sec:gdasrmplus}

In the previous section, we developed an RVU bound for a new variant of predictive regret matching, which closely matches the RVU bound we obtain for $\AdOGD$ (\Cref{sec:AdOGD}). These results are not unrelated; indeed, we claim that they show a deeper relationship between  gradient descent and regret matching. Here, we discuss this relationship. This section will purposefully be  informal and imprecise; its purpose is to give intuition, not formal proofs. 

\subsection{RM+ as adaptive gradient descent}\label{sec:rm-as-ogd}
Our first key remark is that $\RMplus$, with some simple manipulations, already looks quite similar to adaptive gradient descent, where, as in the formal sections, we should think of $\eta^{(t)} := 1/\norm{\vr^{(t)}}_1$ as the ``step size'' at time $t$. Indeed, rearranging the $\RMplus$ update, we have:
\begin{align}
    \vx^{(t+1)}\norm{\vr^{(t+1)}}_1  &= [\vx^{(t)} \norm{\vr^{(t)}}_1 + \vu^{(t)} - \ip{\vu^{(t)}, \vx^{(t)}} \vec 1]^+ \\
    \vx^{(t+1)} &\approx [\vx^{(t)} + \eta^{(t)} (\vu^{(t)} - \ip{\vu^{(t)}, \vx^{(t)}} \vec 1)]^+
\end{align}
where in the second line we approximate $\eta^{(t)} \approx \eta^{(t+1)}$. The operation in the second line is now similar to a projection onto a simplex: indeed, subtracting $\eta^{(t)} \ip{\vu^{(t)}, \vx^{(t)}} \vec 1$ roughly corresponds to projecting into the affine hull of the simplex, and thresholding at zero roughly corresponds to the projection from the affine hull to the simplex itself. For a visualization (without the projection step), see \Cref{fig:rm-as-ogd}. Moreover, the ``step size'' is defined by 
\begin{align}
    \frac{1}{\norm{\vr^{(t)}}_1} \approx \frac{1}{\norm{\vr^{(t)}}_2} \approx \frac{1}{\sqrt{\sum_{\tau=1}^{t-1} \norm{\vg^{(\tau)}}_2^2}}
\end{align}
which is indeed the step size of the non-predictive version of gradient descent.

\begin{figure}
    \centering
\scalebox{0.75}{\begin{tikzpicture}[
    scale=1.2,
    >=Latex, 
    font=\small
]

    \def\Ang{55}       
    \def\Len{5}        
    \def\GLen{3.5}     
    \def\ScaleRatio{0.45} 

    \coordinate (O) at (0,0);
    \coordinate (Rt) at (\Ang:\Len);
    \coordinate (Rt_next) at ($(Rt) + (\Ang-90:\GLen)$);
    \coordinate (L1) at (4,0);
    \coordinate (L2) at (0,4);
    \coordinate (Xt) at (intersection of O--Rt and L1--L2);
    \coordinate (Xt_next) at (intersection of O--Rt_next and L1--L2);

    \draw[thick, dotted, black!80] (4,0) -- (0,4);

    \draw[thin, black!70] (0,0) -- (7,0); 
    \draw[thin, black!70] (0,0) -- (0,6); 

    \begin{scope}[rotate=\Ang]
        \draw[fill=gray!20, draw=black!50] ($(Rt) + (0,0)$) -- ++(0, -0.4) -- ++(-0.4, 0) -- ++(0, 0.4);
    \end{scope}

    \draw[->, red!90!white, thick, line width=1.2pt] (O) -- (Rt) 
        node[anchor=south east, xshift=15pt] {$\bm{r}^{(t)}$};
        
    \draw[->, red!90!white, thick, line width=1.2pt] (O) -- (Rt_next) 
        node[anchor=south west,xshift=-5pt,yshift=-12pt ] {$\bm{r}^{(t+1)}$};

    \draw[->, blue!80!white, thick, line width=1.2pt] (Rt) -- (Rt_next) 
        node[midway, above right, text=blue] {$\bm{g}^{(t)}$};

    \draw[->, blue!80!white, thick, line width=1.2pt] (Xt) -- (Xt_next) 
        node[midway, above right, xshift=-8pt, yshift=-2pt, text=blue] {$\approx \bm{g}^{(t)}/\|\bm{r}^{(t)}\|$};

    \filldraw[fill=red, draw=black] (Xt) circle (2pt) 
        node[above, yshift=3pt, text=red] {$\bm{x}^{(t)}$};
        
    \filldraw[fill=red, draw=black] (Xt_next) circle (2pt) 
        node[right, xshift=3pt, text=red] {$\bm{x}^{(t+1)}$};

    \fill[black] (Rt) circle (1pt);
    \fill[black] (Rt_next) circle (1pt);

\end{tikzpicture}}
    \caption{Visualization of the relationship between $\RM^{(+)}$ and gradient descent, as described in \Cref{sec:rm-as-ogd}. The dotted line indicates the simplex. Because $\vx^{(t)}$ is a factor of $\norm{\vr^{(t)}}$ shorter than $\vr^{(t)}$, when the latter vector moves by $\vg^{(t)}$, the former moves by approximately $\vg^{(t)}/\norm{\vr^{(t)}}$.}
    \label{fig:rm-as-ogd}
\end{figure}

\subsection{Optimistic  gradient descent via infinitesimal steps}

In \Cref{sec:RM}, we discussed an interpretation of our  $\IRPRMplus$ algorithm as using $\RMplus$ as a black-box subroutine using infinitesimal steps during the prediction stage. In this section, we will illustrate a similar argument for predictive gradient descent. Indeed, we will give a proof sketch of an RVU-type bound for predictive gradient descent, using nothing but 1) the adversarial regret bound of (non-optimistic) gradient descent, and 2) elementary geometry. For simplicity, we will consider the unconstrained case. Here, the regret bound of gradient descent (initialized at $\vx^{(1)} = \vec 0$) against a fixed target $\vx^*$ can be written as 
\begin{align}
    \sum_{t=1}^T \ip{\vu^{(t)}, \vx^* - \vx^{(t)}} \le \frac{\norm{\vx^*}_2^2}{2\eta } + \eta \sum_{t=1}^T \norm{\vu^{(t)}}_2^2. \label{eq:gd regret}
\end{align}
Our goal is to use this fact to show the following RVU-type bound.
\begin{proposition}
    The regret of unconstrained optimistic GD, defined by the recurrences
    \begin{align}
        \tilde\vx^{(0)} = \vec 0, \qq{} \vx^{(t)} = \tilde\vx^{(t)} + \eta \vm^{(t)}, \qq{and} \tilde\vx^{(t+1)} = \tilde\vx^{(t)} + \eta \vu^{(t)}
    \end{align}
    is bounded by
    \begin{align}
        \sum_{t=1}^T \ip{\vu^{(t)}, \vx^* - \vx^{(t)}} \le \frac{\norm{\vx^*}_2^2}{2\eta } + \eta \sum_{t=1}^T \norm{\vu^{(t)} - \vm^{(t)}}_2^2  - \frac{1}{2\eta} \sum_{t=1}^T \norm{\vx^{(t)} - \tilde \vx^{(t)}}_2^2.
    \end{align}
\end{proposition}
\begin{proof}[Proof Sketch]
    Optimistic GD accepting utility $\vu^{(t)}$ at time $t$ is iterate-equivalent to (non-optimistic) GD first accepting $\vm^{(t)}/K$ in $K$ steps, then accepting $\vu^{(t)}$ in one step, in the limit $K \to \infty$. By \eqref{eq:gd regret}, the regret of this version of GD is bounded by 
    \begin{align}
    \Reg^{(T)}(\vx^*) &\le \frac{\norm{\vx^*}_2^2}{2\eta } + \eta \sum_{t=1}^T \ab\biggg(\underbrace{\sum_{k=1}^K \norm*{\frac{\vm^{(t)}}{K}}^2}_{\to 0 \text{ as } K \to \infty} + \norm{\vu^{(t)} - \vm^{(t)}}_2^2)
    \\&= \frac{\norm{\vx^*}_2^2}{2\eta } + \eta \sum_{t=1}^T \norm{\vu^{(t)} - \vm^{(t)}}_2^2. \label{eq:ogd via infinitesimals}
\end{align}
This analysis, however, overestimates the regret. By how much? The analysis assumes that, during the $k$th ``small'' step, the learner played $\tilde \vx^{(t, k)} := \vx^{(t)} + \vm^{(t)}$. The average of these vectors, in the limit $K \to \infty$, is $(\vx^{(t)} + \tilde\vx^{(t)})/2$, so during this time the above analysis assumes that the learner obtained utility $\ip{\vm^{(t)}, (\vx^{(t)} + \tilde\vx^{(t)})/2}$. But, in reality, the learner observed the entire utility $\vm^{(t)}$ while playing $\vx^{(t)}$, and so actually obtained utility $\ip{\vm^{(t)}, \vx^{(t)}}$. Thus, the overestimate at time $t$ is
\begin{align}
    \ip*{\vm^{(t)}, \vx^{(t)} - \frac{\vx^{(t)} + \tilde\vx^{(t)}}{2}} = \frac{1}{2\eta} \norm{\vx^{(t)} - \tilde \vx^{(t)}}_2^2.
\end{align}
Summing this overestimate across times $t=1, \dots, T$ and subtracting from \eqref{eq:ogd via infinitesimals} completes the proof.
\end{proof}
One can therefore think of our $\IRPRMplus$ algorithm as arising as the same principle, and analysis technique, applied to regret matching instead of gradient descent.

\section{Constant regret and last-iterate convergence in zero-sum games and beyond}\label{sec:harmonic}

In this section, we show how to use our RVU bounds to recover constant regret and best-iterate convergence in variational inequality problems obeying a property reminiscent of harmonic games. Crucially and unlike prior results on harmonic games, our algorithm, due to its scale-invariance, does {\em not} require knowing {\em a priori} the ``weights'' that define the harmonic game property.  As special cases, we recover these results for zero-sum normal-form games with a scale-invariant version of regret matching---to our knowledge, the first result of this kind. Moreover, due to a novel and curious connection between harmonic games and counterfactual regret minimization (CFR)~\cite{Zinkevich07:Regret}, we recover the same results for extensive-form games in which a fully-mixed equilibrium exists. 

Throughout this section, we will suppress time-independent constant factors with $\lesssim$ and $O_T(\cdot)$. 

\subsection{Constant regret and best-iterate convergence in zero-sum games and beyond}

We will consider problems of the following very general form: let $\cX = \cX_1 \times \dots \times \cX_n$, where $\cX_i \subset \R^{d_i}$ are convex compact sets. Let $F : \cX \to \R^{\sum_i d_i}$ be a Lipschitz operator. The  {\em $\eps$-approximate variational inequality problem}\footnote{Since all our other results are phrased in terms of utility maximization, we follow the utility maximization sign conventions in this section as well.} ($\eps$-VIP) asks for a point $\vx \in \cX$ such that \begin{align}
    \ip{F(\vx), \vx' - \vx} \le \eps \qq{for all} \vx' \in \cX. \label{eq:vi}
\end{align}
We will only consider VIPs with the following additional structure, which we refer to as the {\em weighted Minty property}.

\begin{assumption}[Weighted Minty property]\label{assumption:first}
    There exist weights $w_i > 0$, and some point $\vx^* \in \cX$, such that 
\begin{align}
    \sum_{i = 1}^n w_i \ip{F_i(\vx), \vx^*_i - \vx_i} \ge 0 \qq{for all} \vx \in \cX.
\end{align}
\end{assumption}
Many natural problems fall into this class. For example:
\begin{enumerate}
    \item Monotone VIPs, including zero-sum games, satisfy this property with $w_i = 1$.
    \item Harmonic games~\cite{Candogan11:Flows} satisfy this property with equality instead of inequality.
    \item Extensive-form games with a unique fully-mixed equilibrium satisfy this property in a sense, where the weights $w_i$ correspond to the equilibrium. (We will formalize later the precise sense in which this is true.)
\end{enumerate}
Our analysis in this section will therefore apply to all of these settings. Crucially, our analysis does not depend on knowing the weights {\em a priori}; therefore, it can be used in settings where the weights are unknown. For harmonic games, our results therefore improve upon the results of \citet{Legacci24:No}, which require prior knowledge of the weights. For extensive-form games, since the weights come from the equilibrium, algorithms that require knowing the weights cannot be used.

We will study iterative algorithms of the following generic form. 

\begin{assumption}\label{assumption:second}
The iterate $\vx^{(t)}_i \in \cX_i$ at each timestep $t$ is determined by a regret minimizer $\cR_i$. We assume that each $\cR_i$, on each timestep $t$:
\begin{enumerate}[noitemsep]
    \item produces a pre-iterate $\tilde \vx^{(t)}_i \in \cX_i$, then
    \item accepts a {\em prediction} $\vm^{(t)}_i \in \R^{d_i}$, then
    \item produces an iterate $\vx^{(t)}_i \in \cX_i$, and finally
    \item accepts a {\em utility} $\vu^{(t)}_i \in \R^{d_i}$, constructed as $\vu^{(t)} := F(\vx^{(t)})$.
\end{enumerate}
\end{assumption}

\begin{assumption}[RVU bound]\label{assumption:harmonic-rvu}
    The regret minimizers obey the scale-invariant RVU bound
\begin{align}
    \Reg_i(T) \lesssim \sqrt{\sum_{t=1}^T \norm{\vu^{(t)}_i - \vm^{(t)}_i}_*^2} - \beta \sum_{t=1}^T \norm{\vx^{(t)}_i - \tilde\vx^{(t)}_i}^2
\end{align}
where $\beta > 0$ is a time-independent constant.\footnote{We leave the choice of norm unspecified because all norms on $\R^d$ are equivalent up to constant factors; we will continue to do so throughout this section.}
\end{assumption}

This formalism simultaneously covers two paradigms: the {\em optimistic} learning setup, and the {\em extra-gradient} setup---the former case by simply setting $\tilde\vx^{(t)} := \vx^{(t-1)}$. The setups are formalized in pseudocode in \Cref{alg:learningsetups} in the appendix. With a caveat which we will mention later, by \Cref{theorem:RVU-adOGD}, $\AdOGD$ satisfies \Cref{assumption:harmonic-rvu} in both setups.

Using our results, $\IRPRM^{(+)}$ satisfies the property only in the extra-gradient setup, not the standard optimistic learning setup. This is due to a difference between the two analyses: in the RVU bound for $\IRPRMplus$ (\Cref{th:rvu}), in which the final term is $\vx^{(t)} - \tilde\vx^{(t)}$ instead of $\vx^{(t)} - \vx^{(t-1)}$ in \Cref{def:RVU}; this means that we want to construct the predictions at time $t$ from $\tilde\vx^{(t)}$ instead of $\vx^{(t-1)}$ so that the negative term cancels the positive term, which leads to the extra-gradient setup. We will use $\IREGPRMplus$ to mean running $\IRPRMplus$, in the extra-gradient setup. We leave as an interesting open problem the question of whether similar results can be proven for the optimistic learning setup.

The caveat is the following: in $\AdOGD$, $\beta$ may not be able to be set to a positive number {\em until} the first iteration on which the regret minimizer's prediction is incorrect. Fortunately, the fix is simple: we will insist that a regret minimizer not change its strategy until it selects a sub-optimal strategy (and thus incurs some regret):

\begin{assumption}[No movement until nonzero utility]\label{assumption:nomove}\label{assumption:last}
For every $i$, $\cX_i$ is full-dimensional, and the initial point is $\tilde\vx_i^{(1)} \in \op{int} \cX_i$.    The predictions $\vm_i^{(t)}$ are set as
    \begin{align}
        \vm_i^{(t)} := \begin{cases}
            \displaystyle \vec 0 &\qif \vu_i^{(\tau)} = \vec 0 \quad \forall \tau < t \\
            F_i(\tilde \vx^{(t)}) &\qq{otherwise}
        \end{cases}.
    \end{align}
    Moreover, in the former case we assume $\vx_i^{(t)} = \tilde\vx_i^{(t)} = \tilde\vx_i^{(1)}$.
\end{assumption}
Note that we already incorporated this condition into \Cref{alg:prmp-mod} ($\IRPRMplus$) (and indeed the condition was already used in the proof of \Cref{th:rvu}), so the above assumption holds for $\IRPRMplus$ without further modification.

The full-dimensionality and interior starting point assumptions are mostly for notational simplicity, and are not crucial to the argument. The choice of fixing the predictions to be zero and not changing the strategy until a nonzero utility is accepted {\em is} crucial to the argument.

Under \Cref{assumption:nomove}, for both $\AdOGD$ and $\IRPRM^{(+)}$, we have $\vx_i^{(t)} = \tilde\vx_i^{(t)}$ whenever the step size $\eta^{(t)}_i$ is infinite. Therefore, \Cref{assumption:harmonic-rvu} holds with the setting $$\beta = \min_{i, t : \eta^{(t)}_i < \infty} \frac{1}{\eta^{(t)}_i}$$ where the minimum must exist because the step sizes are non-increasing.

\begin{theorem}[Constant regret]\label{th:fast-convergence}
    Under Assumptions~\ref{assumption:first}--\ref{assumption:last}, the regret of every player is upper-bounded by a time-independent constant.
\end{theorem}
\begin{proof}
For simplicity of notation, we will define the {\em second-order path length} of player $i$ to be
\begin{align}
    P_i^2 := \sum_{t=1}^T \norm{\vx_j^{(t)} - \tilde\vx_j^{(t)}}^2.
\end{align}
For each player $i$, let $t_i^*$ be the first iteration for which $\vu_i^{(t)} \ne \vec 0$. Then in particular, we have
\begin{align}
    \sum_{t=1}^T \norm{\vu^{(t)}_i - \vm^{(t)}_i}_*^2 &= \sum_{t < t_i^*} \norm{\underbrace{\vu^{(t)}_i - \vm^{(t)}_i}_{\clap{\scriptsize $= \vec 0$ by \Cref{assumption:nomove}}}}_*^2  +  \underbrace{\norm{\vu^{(t_i^*)}_i - \vm^{(t_i^*)}_i}_*^2}_{\lesssim 1} + \sum_{t > t_i^*} \norm{\underbrace{\vu^{(t)}_i - \vm^{(t)}_i}_{\clap{\scriptsize $=F_i(\vx^{(t)})-F_i(\tilde \vx^{(t)})$}}}_*^2
    \\&\lesssim 1 + \sum_{t=1}^T \norm{F(\vx^{(t)}) - F(\tilde\vx^{(t)})}_*^2
    \\&\lesssim 1 + \sum_{j=1}^n P_j^2 \label{eq:pathlength}
\end{align}
where the last line follows by continuity (and hence Lipschitzness) of $F$.

By the weighted Minty property, the $w_i$-weighted sum of regrets is nonnegative. Therefore:
\begin{align}
    0 \le \sum_{i=1}^n w_i \Reg_i(T) &\le \sum_{i=1}^n w_i \ab(\sqrt{\sum_{t=1}^T \norm{\vu^{(t)}_i - \vm^{(t)}_i}_*^2} - \beta P_i^2).
    \\&\lesssim \sum_{i=1}^n  w_i \ab(\sqrt{1 + \sum_{j=1}^n P_j^2} - \beta P_i^2).
    \\&\lesssim 1 + \sum_{i=1}^n (P_i - \beta w_i P_i^2)
    \\&\lesssim 1 -  \sum_{i=1}^n \frac{\beta w_i}{2} P_i^2
\end{align}
where we use, in turn: the definition of weighted Minty, the RVU bound, \eqref{eq:pathlength}, the inequality $\sqrt{a + b} \le \sqrt{a} + \sqrt{b}$ for $a, b \ge 0$, and finally the inequality $x - ax^2 \le 1/(2a) - (a/2) x^2$, which holds for any $a > 0$.
Thus we have $\sum_{i=1}^n P_i^2 \lesssim 1$, and combining this with \eqref{eq:pathlength} and the first term of the RVU bound completes the proof.
\end{proof}

We now move to best-iterate convergence. To show this, we will need an additional assumption.

\begin{assumption}[One-step improvement]\label{assumption:onestep}
    The regret minimizers obey the property
\begin{align}
     \norm{\vx_i^{(t)} - \tilde\vx_i^{(t)}} 
     \gtrsim \frac{\max_{\vx'_i \in \cX_i} \ip{\vm^{(t)}_i, \vx_i' - \tilde\vx_i^{(t)}}}{\sqrt{\sum_{t=1}^T \norm{\vu^{(t)}_i - \vm^{(t)}_i}_*^2}}.
\end{align}
\end{assumption}
In the appendix (\Cref{lem:onesteprmplus,lem:gd-onestepimprov}) we show that $\AdOGD$ and $\IRPRMplus$ satisfy this property. $\IRPRM$, however, {\em does not} satisfy one-step improvement, and therefore best-iterate convergence does {\em not} hold for $\IRPRM$. This should not be surprising; indeed, if we use the analogy that $\RM$ is to FTRL what $\RMplus$ is to OMD, then the former should fail to have fast last-iterate convergence, for the same reason that FTRL in general fails to have fast convergence~\cite{Cai24:Fast}.

\begin{theorem}[Best-iterate convergence]
    \label{thm:fast-bestiterate}
    Under Assumptions~\ref{assumption:first}--\ref{assumption:last} and \ref{assumption:onestep}, there exists an iteration $t \le T$ for which $\vx^{(t)}$ is an $O_T(1/\sqrt{T})$-VIP solution.
\end{theorem}
\begin{proof}
Assume that, for every player $i$, either $t_i^* = \infty$ or $T > 2t_i^*$, so that $\vm^{(t)} = F(\tilde\vx^{(t)})$. (This is valid because $t_i^*$ does not depend on $T$.)
    From the proof of \Cref{th:fast-convergence}, we have $\sum_{i=1}^n P_i^2 \lesssim 1$. Thus, there is an iteration $t \in [T/2, T]$ on which 
    \begin{align}
       \sum_{i=1}^n \norm{\vx_i^{(t)} - \tilde\vx_i^{(t)}}^2\lesssim \frac{1}{T}. \label{eq:bestiterate}
    \end{align}
    Combining the above with \eqref{eq:pathlength}, \Cref{th:fast-convergence}, \Cref{assumption:onestep}, and $\vm^{(t)} = F(\tilde\vx^{(t)})$, we have
    \begin{align}
        \max_{\vx'_i \in \cX_i} \ip{F_i(\tilde\vx^{(t)}), \vx_i' - \tilde\vx_i^{(t)}} \lesssim \frac{1}{\sqrt{T}},
    \end{align}
    that is, $\tilde\vx^{(t)}$ is an $O_T(1/\sqrt{T})$-VIP solution. Moreover, by \eqref{eq:bestiterate}, we have $\norm{\vx^{(t)} - \tilde\vx^{(t)}} \lesssim 1/\sqrt{T}$, so by Lipschitzness of $F$ it follows that $\vx^{(t)}$ is also an $O_T(1/\sqrt{T})$-VIP solution.
\end{proof}

\section{Constant regret and best-iterate convergence for CFR in zero-sum extensive-form games}
In this section, we show how to apply our results to derive the first constant regret and best-iterate convergence results for {\em counterfactual regret minimization} (CFR) with scale-invariant local regret minimizers in extensive-form games. To do this, we will first require some background about extensive-form games.

A {\em tree-form decision problem} is a tree with two types of nodes---{\em decision points} $j \in \cJ$ and {\em observation points} $\sigma \in \Sigma$ (also known as {\em sequences})---modeling a multi-timestep decision problem for a single player. We assume without loss of generality that decision points and observation points alternate down every path in the tree, and that leaves (terminal nodes) are all observation points. At decision points $j$, the player selects one of the child observation points; by convention, each outgoing edge is labeled with an {\em action} $a$, and the observation point reached by following action $a$ at $j$ is denoted $ja$. The set of legal actions at decision point $j$ is denoted $\cA_j$. The set of decision point children of a sequence $\sigma$ is denoted $C_\sigma$, and conversely for a decision point $j$ the parent sequence is denoted $p_j$. A {\em (behavioral) strategy} $\vb$ is therefore an assignment of a probability distribution $\vx_j \in \Delta(\cA_j)$ to each decision point $j \in \cJ$. 

The {\em sequence form} of a behavioral strategy $\vx$ is the vector $\vq\in \R^\Sigma$, indexed by sequences, for which $\vq[\sigma]$ is the probability that the player plays all actions on the path to $\sigma$:
\begin{align}
    \vq[\sigma] = \prod_{ja \preceq \sigma} \vx_j[a]
\end{align}
where $\preceq$ denotes the ordering over the tree. The set of sequence-form strategies, which we will denote $\cQ$ for each player $i$, is known to be a convex compact set~\cite{Romanovskii62:Reduction,Koller94:Fast,Stengel96:Efficient}.

For our purposes, a two-player zero-sum extensive-form game is a special case of a monotone VIP on $\cQ$.\footnote{$\cQ$ here would actually be the product $\cQ_1 \times \cQ_2$ of the two players' strategy sets, but since the product of sequence-form strategy sets is also (up to affine transformation) a sequence-form strategy set, it is fine for us to think of the joint strategy set $\cQ$ directly, and ignore the two players.} That is, the VIP associated with such a game has the form $G : \cQ \to \R^{\Sigma}$ where $G$ is Lipschitz and satisfies the monotonicity condition
\begin{align}
\ip{G(\vq) - G(\vq'), \vq - \vq'} \le 0 \quad \forall \vq, \vq' \in \cQ. 
\end{align}
Our goal will be to solve this VIP. However, perhaps surprisingly, we will do so not directly; instead, we will construct a {\em different} VIP over the {\em behavioral} strategy set $\cX := \bigtimes_j \Delta(\cA_j)$.

{\em Counterfactual regret minimization (CFR)}~\cite{Zinkevich07:Regret} is a popular framework for constructing regret minimizers over sequence-form strategy sets, and therefore also for solving extensive-form games. It consists of initializing an independent regret minimizer $\cR_j$ for each decision point $j$, which decides on the behavioral strategy $\vx^{(t)}_j$ to be played at every timestep $t$. The utility vector observed at decision point $j$ is the {\em counterfactual utility}, defined recursively by the formula
\begin{align}
    F_j(\vx)[a] = G(\vq)[ja] + \sum_{k \in C_{ja}} \ip{F_k(\vx), \vx_{k}}
\end{align}
where $\vq$ is the sequence form of $\vx$. The choice of notation is deliberate: the counterfactual regret minimization algorithm is equivalent to running regret minimization on the VIP $F : \cX \to \R^\Sigma$.

The crucial property of CFR is that the regret {\em in sequence form} is a certain weighted sum of regrets of the individual regret minimizers {\em in behavioral form}. In our notation, this can be expressed as follows.
\begin{restatable}[\citealp{Zinkevich07:Regret,Farina19:Regret}]{proposition}{propCFR}\label{prop:cfr}
    Let $\vx^{(1)}, \dots, \vx^{(T)} \in \cX$ be a sequence of behavioral strategies and $\vq^{(t)}$ be the sequence form of $\vx^{(t)}$ for each $t$. Then for any $\vq^* \in \cX$, we have
    \begin{align}
    \sum_{t=1}^T \ip{G(\vq^{(t)}), \vq^* - \vq^{(t)}} \le \sum_{j \in \cJ} \vq^*[p_j] \cdot \Reg_j(T)
    \end{align}
\end{restatable}
In particular, if $\vq^*$ is a Nash equilibrium, the left-hand side is nonnegative. Moreover, call $\vq \in \cQ$ {\em full-reach} if every decision point is reached with positive probability, that is, $\vq[p_j] > 0$ for every $j$. If a full-reach Nash equilibrium exists, then by the above proposition, the weighted Minty condition holds for the VIP $F$, where the weights are given by $w_j = \vq^*[p_j]$. 

Thus, by the results in \Cref{sec:harmonic}, we have:
\begin{theorem}\label{th:efg}
    Suppose that we run CFR on a two-player zero-sum extensive-form game that admits a full-reach Nash equilibrium, with regret minimizers satisfying Assumptions~\ref{assumption:second}--\ref{assumption:last} at each decision point $j$. Then we achieve {\em constant regret}: all the counterfactual regrets $\Reg_j(T)$ are bounded by a time-independent constant, and the average iterate $\bar\vq := \frac{1}{T} \sum_{t=1}^T \vq^{(t)}$ is an $O_T(1/T)$-approximate Nash equilibrium.

    If moreover the regret minimizers satisfy \Cref{assumption:onestep}, then we also achieve {\em best-iterate convergence}: there exists a timestep $t \le T$ at which $\vq^{(t)}$ is an $O_T(1/\sqrt{T})$-approximate Nash equilibrium.
\end{theorem}

\begin{corollary}\label{cor:iregpcfr}
    Suppose that we run CFR on a two-player zero-sum extensive-form game that admits a full-reach Nash equilibrium, with $\AdOGD$ in either the optimistic or extra-gradient setup, or with $\IREGPRMplus$, as the local regret minimizers.  Then we achieve constant regret and $O_T(1/\sqrt{T})$ best-iterate convergence.
\end{corollary}
In particular, fully-mixed strategies are full-reach, and in Bayesian games, all strategies are full-reach; therefore, the above corollaries apply to all zero-sum extensive-form games with fully-mixed equilibria, and all zero-sum Bayesian (and thus all normal-form) games regardless of whether a fully-mixed equilibrium exists.

The above result is worth noting for two reasons. First, as mentioned before, to our knowledge it is the first result that shows best-iterate convergence and constant regret bounds in zero-sum games with scale-invariant local regret minimizers, even for normal-form games. Second, the proof technique yields an intriguing connection between CFR and the Minty condition,  which can be restated in game-theoretic terms as follows: if one views a zero-sum extensive-form game as an {\em agent-form game} wherein each decision point is its own player the utilities are given by the counterfactual utilities, then 1) running no-regret learning on this game yields CFR, and 2) the existence of a fully-mixed Nash equilibrium implies the weighted Minty condition for this game, which in turn implies fast convergence.

Throughout the previous two sections, we have been extremely liberal with our use of $\lesssim$ to hide game-dependent constants, including: $\poly(n, d)$ factors; the constant $\beta$; the inverse weights $1/w_i$, which for extensive form corresponds to the minimum reach probability of any information set in the equilibrium $\vq^*$; and the critical timesteps $t_i^*$. We leave it to future research to determine whether these factors can be removed, in particular whether one can achieve bounds depending only on $\poly(n, d)$ factors, or for extensive-form games without full-reach equilibria.

\section{Experiments}
\label{sec:exp}

\begin{figure}
\newcommand\widthfactor{0.45}
\newcommand\plottype{loglog}
\newcommand\includegame[1]{\includegraphics[width=\widthfactor\textwidth]{experiments/#1_\plottype.pdf}}
    \centering
    \includegame{counterexample.txt}
    \includegame{games_ld_liteefg}
    \includegame{games_kuhn}
    \includegame{games_leduc_liteefg}
    \includegame{games_correlation_battleship_3x2}
    \includegame{games_goofspiel4i}
    \includegraphics[scale=0.7]{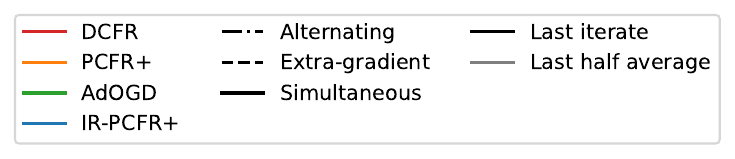}
    \caption{Experimental results. The $x$-axis is the number of gradient evaluations (matrix-vector products with $\mA$): alternating and simultaneous iterates use two gradient evaluations per iteration; extra-gradient uses four. $\DCFR$ is not typically run with predictions, so we also do not use predictions when running $\DCFR$, and thus ``Extra-gradient $\DCFR$'' is not run. To avoid messy plots, the average iterate is only shown if it is better than the last iterate, and only the lower frontier of each curve is shown, that is, each curve plots the smallest Nash gap achieved up to that timestep.}
    \label{fig:experiments}
\end{figure}

We ran experiments on various extensive-form games commonly used as benchmarks in the literature. We tested four algorithms: $\DCFR$~\cite{Brown19:Solving}, $\PRMplus$, $\AdOGD$, and $\IRPRMplus$. These algorithms were run at every information set independently using the $\CFR$ framework~\cite{Zinkevich07:Regret}; therefore, we will refer to $\PRMplus$ and $\IRPRMplus$ as $\PCFRplus$ and $\IRCFRplus$ respectively for this section. For each algorithm, we tested three setups: simultaneous iterates, alternating iterates, and extragradient. We recorded the Nash gap of both the last iterate and the average of the most recent half of iterates. All experimental results are in \Cref{fig:experiments}. 
The games are as follows.
\begin{itemize}
    \item {\bf Farina et al. Counterexample}---the normal-form game \eqref{eq:counterexample}~\cite{Farina23:Regret}.
    \item {\bf Liar's dice, Kuhn poker}, and {\bf Leduc poker}---standard games, as found in, for example, LiteEFG~\cite{Liu24:LiteEFG}.
    \item A version of {\bf Goofspiel}~\cite{Lanctot09:Monte}, with 4 cards per player, imperfect information, and a fixed deck order.
    \item A version of {\bf Battleship}, with 2 turns per player on a 2x3 board and a single ship of length 2.
\end{itemize}
We make several observations about the experimental results.

\paragraph{Selective superiority.} There is no algorithm that is consistently best across all games.

\paragraph{Linear last-iterate convergence.} All algorithms tested, except $\DCFR$, $\PCFRplus$, and extragradient $\PCFRplus$, appear to consistently exhibit {\em linear} last-iterate convergence. This phenomenon, especially in extensive-form games, is unexplained theoretically, especially in extensive-form games, and is an interesting topic of future research. Due to this linear convergence, most other algorithms eventually overtake $\DCFR$ in the high-precision regime, with $\DCFR$ only remaining slightly superior in average iterate on a single game (Leduc poker). 

\paragraph{Alternation.} As is well known in the literature, using alternation is better than not using alternation in practice. That remains true in our experiments. However, our algorithms $\AdOGD$ and $\IRCFRplus$ significantly close this gap: their simultaneous variants, unlike simultaneous $\PCFRplus$, appear to converge in iterates, and at rates not significantly behind, or even occasionally slightly faster than, the alternating variants. 

\paragraph{Per-iterate time complexity.} (Not shown in graphs.) $\PCFRplus$ and $\DCFR$ are simple algorithms, requiring only a few vectorizable operations per information set per iteration. They hence are very fast per-iterate. $\IRCFRplus$, while still linear time per iteration, requires a substantially more complex computation (see \Cref{sec:gamma}), and is therefore slower per iteration in practice. $\AdOGD$ requires a projection onto the simplex, which has similar time complexity~\cite{Condat16:Fast}.

\paragraph{Scale invariance.} \citet{Chakrabarti24:Extensive} hypothesized that the property that makes $\PCFRplus$ a powerful practical algorithm is local---that is, information set-level---{\em scale invariance}. Our results support this hypothesis. In our view, there is not much remaining that is ``special'' about $\PCFRplus$, and its powerful practical performance is explained by the fact that it is performing gradient-descent-like updates using the ``theoretically optimal'' step size of (at least) $1/\sqrt{P^{(t)}}$. Indeed, our experimental results support this view: gradient descent, with the correct adaptive step size of $1/\sqrt{P^{(t)}}$, performs similarly to $\PCFRplus$.

\section{Conclusion and future research}
There has long been a mystery about why $\RMplus$ performs so well in practice, especially when compared to other algorithms such as $\OGD$ which had better theoretical guarantees. In this paper, we have made a significant step toward solving this mystery, from both directions. We devised a variant of $\PRMplus$, and an adaptive learning rate variant of $\OGD$, $\AdOGD$. Both algorithms maintain the theoretical $O_T(1/T)$ average-iterate and $O_T(1/\sqrt{T})$ best-iterate convergence rates of $\OGD$, while additionally gaining the scale-invariance property that seems to make $\RMplus$ powerful in practice. In experiments, all three algorithms have similar properties and performance, including fast last-iterate convergence at seemingly linear rates.

Many interesting questions remain for future research.
\begin{enumerate}
    \item What properties can be proven about the alternating variants of these algorithms, especially $\PCFRplus$?
    \item Does $\IRPRMplus$ have a best-iterate and/or $O_T(1/T)$ convergence rate when used {\em without} the extra-gradient setup (\ie, in the usual simultaneous iterate learning setup)? In \Cref{sec:RM} we discussed the steps that would be required to show this.
    \item Can one show a $\poly(m, n)/T$ average-iterate convergence rate (or $\poly(m, n)/\sqrt{T}$ best-iterate) for $\AdOGD$ or $\IREGPRM^{(+)}$? Our current bounds depend on several game-dependent constants (hidden by $\lesssim$ throughout our analyses).
    \item Our convergence results in extensive-form games hinge on having a full-reach Nash equilibrium. It is an interesting future direction to extend them to any zero-sum extensive-form game.
    \item Many of these algorithms exhibit {\em linear} last-iterate convergence rates in practice. Is linear last-iterate theoretically guaranteed for any or all of these algorithms?
\end{enumerate}

\section*{Acknowledgements}

We thank Haipeng Luo and Gabriele Farina for many insightful conversations.
T.S. is supported by the Vannevar Bush Faculty Fellowship ONR N00014-23-1-2876, National Science Foundation grants RI-2312342 and RI-1901403, ARO award W911NF2210266, and NIH award A240108S001.

\bibliographystyle{plainnat}
\bibliography{dairefs}

\newpage
\appendix

\section{Adaptive optimistic gradient descent}
\label{sec:AdOGD}

As a warm-up to the analysis in~\Cref{sec:RM}, here we analyze the usual optimistic mirror descent algorithm~\citep{Rakhlin13:Online} with Euclidean regularization, but with a particular type of time-varying learning rate; we call the resulting algorithm $\AdOGD$. As discussed in~\Cref{sec:rm-as-ogd}, there are many parallels between this adaptive gradient descent-type algorithm and $\IREGPRMplus$---the algorithm that we introduced in~\Cref{sec:RM}.

The theory we develop in this section applies to a general convex and compact set $\cX$. In this context, $\AdOGD$ is defined as follows. We first initialize $\cX \ni \tvx^{(1)} = \vx^{(1)} \in \argmax_{\vx \in \cX} \langle \vx, \vm^{(1)} \rangle$. Then, for $t = 1, \dots, T$,
\begin{equation}
    \label{eq:AOGD}
    \tag{\texttt{AdOGD}}
\begin{aligned}    
    \Tilde{\vx}^{(t+1)} \defeq \argmax_{\Tilde{\vx} \in \cX} \left\{ \eta^{(t)} \langle \Tilde{\vx}, \vu^{(t)} \rangle - \frac{1}{2} \| \tilde{\vx} - \Tilde{\vx}^{(t)}  \|_2^2 \right\} = \proj_{\cX}( \tvx^{(t)} + \eta^{(t)} \vu^{(t)} ),\\
    \vx^{(t+1)} \defeq \argmax_{\vx \in \cX} \left\{ \eta^{(t+1)} \langle \vx, \vm^{(t+1)} \rangle - \frac{1}{2} \| \vx - \Tilde{\vx}^{(t+1)}  \|_2^2 \right\} = \proj_\cX(\tvx^{(t+1)} + \eta^{(t+1)} \vm^{(t+1)}).
\end{aligned}
\end{equation}
Above, $\proj_\cX$ denotes the Euclidean projection to $\cX$ and $(\eta^{(t)})_{t=1}^T$ is the learning rate sequence, which is to be tuned appropriately (\Cref{theorem:RVU-adOGD}). By convention, if $\eta^{(t)} = +\infty$ in the proximal step of $\tvx^{(t+1)}$, we take $\tvx^{(t+1)}$ to be a best response to $\vu^{(t)}$ with respect to some consistent tie-breaking rule; the same applies to~$\vx^{(t+1)}$.

The first step is to prove an RVU-type bound parameterized on the learning rate sequence. As we shall see, the key precondition to carry out the analysis is that the learning rate is nonincreasing, which, when equating the learning rate to the inverse of the norm of the regret vector, amounts to insisting on having a nondecreasing regret vector. Maintaining this invariance was indeed crucial in~\Cref{sec:RM}, underpinning the basic idea behind $\IRPRMplus$.

In what follows, we denote by $\utilbound$ an upper bound on $\| \vu - \vu' \|_2 $ for all $\vu, \vu' \in \cU$, where $\cU$ is the set of allowable utilities such that $\vec{0} \in \cU$. We always assume that the prediction vector satisfies $\vm^{(t)} \in \cU$, which holds, for example, when we set $\vm^{(t)} = \vu^{(t-1)}$.

\begin{theorem}[RVU bound for $\AdOGD$]
    \label{theorem:RVU-adOGD}
    For any nonincreasing learning rate sequence, the regret $\max_{\vx^* \in \cX} \sum_{t=1}^T \langle \vx^* - \vx^{(t)}, \vu^{(t)} \rangle$ of $\AdOGD$ can be upper bounded by
    \begin{equation}
        \label{eq:first-ub}
        \frac{\diam^2_{\cX}}{\eta^{(T)}} + \sum_{t=1}^T \eta^{(t)} \|\vu^{(t)} - \vm^{(t)}\|_2^2 - \sum_{t=1}^T \frac{1}{2 \eta^{(t)}} \|\vx^{(t)} - \tilde{\vx}^{(t)} \|_2^2 - \sum_{t=1}^T \frac{1}{2 \eta^{(t)}} \|\vx^{(t)} - \tilde{\vx}^{(t+1)} \|_2^2.
    \end{equation}
    In particular, if $\delta = \| \vu^{(1)} - \vm^{(1)} \|_2 > 0$, $\mispred^{(t)} \defeq \sum_{\tau=1}^{t-1} \| \vu^{(\tau)} - \vm^{(\tau)} \|_2^2$, and $\eta^{(t)} \defeq \etao/\sqrt{\mispred^{(t)}}$ for $t \geq 2$ and $\eta^{(1)} = \eta^{(2)}$, \eqref{eq:first-ub} can be in turn upper bounded by
    \begin{equation}
        \label{eq:modified-RVU}
        \left( 3 \etao \frac{\utilbound}{\delta}  + \frac{\diam_{\cX}^2}{\etao} \right) \sqrt{ \sum_{t=1}^T \| \vu^{(t)} - \vm^{(t)} \|_2^2 }  - \frac{\delta}{2\etao} \left( \sum_{t=1}^T \|\vx^{(t)} - \tilde{\vx}^{(t)} \|_2^2 + \sum_{t=1}^T  \|\vx^{(t)} - \tilde{\vx}^{(t+1)} \|_2^2 \right).
    \end{equation}
\end{theorem}
A few remarks are in order. First, $\diam_\cX$ denotes the maximum between the $\ell_2$-diameter of $\cX$ and $\max_{\vx \in \cX} \|\vx\|_2$. The regret bound in~\eqref{eq:first-ub} closely matches the RVU bound per~\Cref{def:RVU}, with the difference that the underlying parameters are time-varying. For completeness, we carry out the analysis by incorporating a hyperparameter $\eta$ in the definition of the learning rate sequence, but one can take $\eta = 1$ without qualitatively affecting our bounds. The regret bound in~\eqref{eq:modified-RVU} is also a modified RVU-type bound. It depends on the misprediction error after the first round, denoted by $\delta$, which is assumed to be strictly positive; this is without any essential loss: as long as the predictions are perfectly accurate, the algorithm will incur constant regret (in fact, zero), while one can employ the analysis of~\Cref{theorem:RVU-adOGD} when and if a prediction is even slightly inaccurate.%

\begin{proof}[Proof of~\Cref{theorem:RVU-adOGD}]
    By $1$-strong convexity of each of the proximal steps in $\AdOGD$, we have that for any $\tilde{\vx} \in \cX$ and $t \geq 1$,
    \begin{equation}
        \label{eq:sc-1}
        \eta^{(t)} \langle \tilde{\vx}^{(t+1)}, \vu^{(t)} \rangle
    - \frac{1}{2} \| \tilde{\vx}^{(t+1)} - \tilde{\vx}^{(t)} \|_2^2
    - \eta^{(t)} \langle \tilde{\vx}, \vu^{(t)} \rangle
    + \frac{1}{2} \| \tilde{\vx} - \tilde{\vx}^{(t)} \|_2^2
    \geq
    \frac{1}{2} \| \tilde{\vx} - \tilde{\vx}^{(t+1)} \|_2^2.
    \end{equation}
    Similarly, for any $\vx \in \cX$ and $t \geq 2$,
    \begin{equation}
        \label{eq:sc-2}
            \eta^{(t)} \langle \vx^{(t)}, \vm^{(t)} \rangle 
    - \frac{1}{2} \| \vx^{(t)} - \tilde{\vx}^{(t)} \|_2^2 
    - \eta^{(t)} \langle \vx, \vm^{(t)} \rangle 
    + \frac{1}{2} \| \vx - \tilde{\vx}^{(t)} \|_2^2 
    \geq 
    \frac{1}{2} \| \vx - \vx^{(t)} \|_2^2.
    \end{equation}
    By definition of $\vx^{(1)} = \tvx^{(1)}$, \eqref{eq:sc-2} also holds for $t = 1$. Now, for any $\vx^* \in \cX$, we have $\langle \vx^* - \vx^{(t)}, \vu^{(t)} \rangle 
= \langle \vu^{(t)} - \vm^{(t)}, \tilde{\vx}^{(t+1)} - \vx^{(t)} \rangle 
+ \langle \tilde{\vx}^{(t+1)} - \vx^{(t)}, \vm^{(t)} \rangle 
+ \langle \vx^* - \tilde{\vx}^{(t+1)}, \vu^{(t)} \rangle$. Adding~\eqref{eq:sc-1} for $\tilde{\vx} = \vx^*$ and~\eqref{eq:sc-2} for $\vx = \tilde{\vx}^{t+1}$, 
    \begin{align}
        \eta^{(t)} \langle \vx^* - \tilde{\vx}^{(t+1)}, \vu^{(t)} \rangle 
    + \eta^{(t)} \langle \tilde{\vx}^{(t+1)} - \vx^{(t)}, \vm^{(t)} \rangle 
    &\leq 
    \frac{1}{2} \| \vx^* - \tilde{\vx}^{(t)} \|_2^2 
    - \frac{1}{2} \| \vx^* - \tilde{\vx}^{(t+1)} \|_2^2 
   \\ &- \frac{1}{2} \| \vx^{(t)} - \tilde{\vx}^{(t)} \|_2^2 
    - \frac{1}{2} \| \tilde{\vx}^{(t+1)} - \vx^{(t)} \|_2^2
    \end{align}
    Furthermore,
    \begin{align}
    \sum_{t=1}^T \left( \frac{1}{2\eta^{(t)}} \|\vx^* - \tilde{\vx}^{(t)} \|_2^2 
    - \frac{1}{2 \eta^{(t)}} \|\vx^* - \tilde{\vx}^{(t+1)} \|_2^2 \right) 
    &\leq \frac{1}{2\eta^{(1)}} \|\vx^* - \tilde{\vx}^{(1)} \|_2^2 \\
    &+ \sum_{t=1}^{T-1} \|\vx^* - \tilde{\vx}^{(t+1)} \|_2^2 
    \left( \frac{1}{2\eta^{(t+1)}} - \frac{1}{2\eta^{(t)}} \right) \\
    &\leq \frac{1}{2\eta^{(1)}} \|\vx^* - \tilde{\vx}^{(1)} \|_2^2 
    + \diam_{\cX}^2 \sum_{t=1}^{T-1} \left( \frac{1}{2\eta^{(t+1)}} - \frac{1}{2\eta^{(t)}} \right) \\
    &\leq \diam_{\cX}^2 \left( \frac{1}{2\eta^{(1)}} + \frac{1}{2\eta^{(T)}} \right) 
    \leq \frac{\diam_{\cX}^2}{\eta^{(T)}} ,
\end{align}
where we used that $\eta^{(t+1)} \leq \eta^{(t)}$ for all $t$. To bound $\langle \vu^{(t)} - \vm^{(t)}, \tilde{\vx}^{(t+1)} - \vx^{(t)} \rangle$, we add~\eqref{eq:sc-1} for $\tvx = \vx^{(t)}$ and~\eqref{eq:sc-2} for $\vx = \tvx^{(t+1)}$, which implies $\|\vx^{(t)} - \tvx^{(t+1)} \|_2 \leq \eta^{(t)} \|\vu^{(t)} - \vm^{(t)} \|_2$. So, $\langle \vu^{(t)} - \vm^{(t)}, \tilde{\vx}^{(t+1)} - \vx^{(t)} \rangle \leq \eta^{(t)} \|\vu^{(t)} - \vm^{(t)} \|_2^2$. This completes the first part of the proof.
    
    For the second part, we observe that, by the AM-GM inequality,
    \begin{equation}
        \label{eq:first-telescop}
        \frac{\| \vu^{(t)} - \vm^{(t)} \|_2^2 }{\sqrt{ \mispred^{(t+1)} }} = \frac{ \mispred^{(t+1)} - \mispred^{(t)}}{\sqrt{\mispred^{(t+1)}}} \leq 2 \sqrt{\mispred^{(t+1)}} - 2 \sqrt{\mispred^{(t)}}.
    \end{equation}
    Further, $\mispred^{(t+1)} \leq \mispred^{(t)} + \utilbound^2 $, which implies
    \begin{equation}
        \label{eq:secondbound}
        \frac{\mispred^{(t+1)} }{\mispred^{(t)}} \leq 1 + \frac{\utilbound^2}{\delta^2} \leq 2 \frac{\utilbound^2}{\delta^2}
    \end{equation}
    since $P^{(t)} \geq \delta^2$ and $\utilbound \geq \delta$. Combining~\eqref{eq:first-telescop} and~\eqref{eq:secondbound},
    \begin{equation}
        \frac{\| \vu^{(t)} - \vm^{(t)} \|_2^2 }{\sqrt{ \mispred^{(t)} }} \leq \sqrt{2} \frac{\utilbound}{\delta} \frac{ \mispred^{(t+1)} - \mispred^{(t)}}{\sqrt{P^{(t+1)}}} \leq 3 \frac{\utilbound}{\delta} \left( \sqrt{\mispred^{(t+1)}} - \sqrt{\mispred^{(t)}} \right)
    \end{equation}
    for all $t \geq 2$. For $t = 1$, a bound on $\eta^{(t)} \|\vu^{(t)} - \vm^{(t)} \|_2^2$ follows directly from~\eqref{eq:first-telescop}. The claim now follows from a telescopic summation.
\end{proof}

\Cref{theorem:RVU-adOGD} applies under any sequence of utilities. We now use it to show that when both players in a zero-sum game employ $\AdOGD$, their average strategies converge at a rate of $T^{-1}$ to a minimax equilibrium. In what follows, we define
\begin{align}
    L \defeq \max\ab\{ \sup_{\vx, \vx' \in \cX} \frac{\norm{\mA^\top \vx - \mA^\top \vx'}_2}{\norm{\vx - \vx'}_2}, \sup_{\vy, \vy' \in \cY} \frac{\norm{\mA\vy - \mA \vy'}_2}{\norm{\vy - \vy'}_2} \}.\label{eq:L}
\end{align}

\begin{corollary}
    \label{cor:opt-AdOGD}
    Let $\vm_\cX^{(t)} = \vu_\cX^{(t-1)}$ for $t \geq 2$ and $\vm_\cX^{(1)} = \vec{0}$, and similarly for Player $\cY$. If both players employ $\AdOGD$ per~\Cref{theorem:RVU-adOGD} and $\delta_\cX =  \| \vu^{(1)}_\cX \|_2 > 0, \delta_\cY =  \| \vu^{(1)}_\cY \|_2 > 0$, the duality gap of $(\bvx^{(T)}, \bvy^{(T)})$ is bounded by
    \begin{equation}
        \frac{1}{T} \left( \beta_\cX(\etao_\cX) \diam_\cX + \beta_\cY(\etao_\cY) \diam_\cY + \frac{\beta^2_\cX(\etao_\cX)}{4 \alpha_\cY(\etao_\cY)} + \frac{\beta^2_\cY(\etao_\cY)}{4 \alpha_\cX(\etao_\cX)} \right),
\end{equation}
where $\beta_\cX = \left( 3 \etao_\cX \frac{L^2 \diam_\cX}{\delta_\cX}  + \frac{L \diam_{\cX}^2}{\etao_\cX} \right)$, $\beta_\cY = \left( 3 \etao_\cY \frac{L^2 \diam_\cY}{\delta_\cY}  + \frac{L \diam_{\cY}^2}{\etao_\cY} \right)$, $\alpha_\cX = \frac{\delta_\cX}{8 \etao_\cX}$, and $\alpha_\cY = \frac{\delta_\cY}{8 \etao_\cY}$.
\end{corollary}
In the statement above, $\eta_\cX$ and $\eta_\cY$ serve the role of $\eta$ (in accordance with~\Cref{theorem:RVU-adOGD}) for Player $\cX$ and $\cY$, respectively; in what follows, one can take $\eta_\cX = 1 = \eta_\cY$.

\begin{proof}[Proof of~\Cref{cor:opt-AdOGD}]
    Applying~\Cref{theorem:RVU-adOGD} for Player $\cX$,
    \begin{align}
        \reg^{(T)}_\cX \leq \left( 3 \etao_\cX \frac{L^2 \diam^2_\cX}{\delta_\cX} + \frac{L \diam_{\cX}^3}{\etao_\cX} \right) + \left( 3 \etao_\cX \frac{L^2 \diam_\cX}{\delta_\cX}  + \frac{L \diam_{\cX}^2}{\etao_\cX} \right) \sqrt{ \sum_{t=2}^T \| \vy^{(t)} - \vy^{(t-1)} \|_2^2 }  \\
        - \frac{\delta_\cX}{2\etao_\cX} \left( \sum_{t=1}^T \|\vx^{(t)} - \tilde{\vx}^{(t)} \|_2^2 + \sum_{t=1}^T  \|\vx^{(t)} - \tilde{\vx}^{(t+1)} \|_2^2 \right),
    \end{align}
    where we used that $\utilbound_\cX \leq L \diam_\cX$. In particular,
    \begin{align}
        \reg_\cX^{(T)} &\leq \left( 3 \etao_\cX \frac{L^2 \diam^2_\cX}{\delta_\cX} + \frac{L \diam_{\cX}^3}{\etao_\cX} \right) + \left( 3 \etao_\cX \frac{L^2 \diam_\cX}{\delta_\cX}  + \frac{L \diam_{\cX}^2}{\etao_\cX} \right) \sqrt{ \sum_{t=2}^T \| \vy^{(t)} - \vy^{(t-1)} \|_2^2 } \\
        &- \frac{\delta_\cX}{8 \etao_\cX} \sum_{t=2}^T \|\vx^{(t)} - \vx^{(t-1)} \|_2^2 -\frac{\delta_\cX}{4 \etao_\cX} \left( \sum_{t=1}^T \|\vx^{(t)} - \tilde{\vx}^{(t)} \|_2^2 + \sum_{t=1}^T  \|\vx^{(t)} - \tilde{\vx}^{(t+1)} \|_2^2 \right),\label{align:first-regx}
    \end{align}
    where we used that $\|\vx^{(t)} - \vx^{(t-1)} \|_2^2 \leq 2 \| \vx^{(t)} - \tvx^{(t)} \|_2^2 + 2 \|\tvx^{(t)} - \vx^{(t-1)} \|_2^2$, which implies
    \begin{align}
        \sum_{t=2}^T \|\vx^{(t)} - \vx^{(t-1)} \|_2^2 &\leq 2 \sum_{t=2}^T \|\vx^{(t)} - \tilde{\vx}^{(t)} \|_2^2 + 2 \sum_{t=1}^{T-1}  \|\vx^{(t)} - \tilde{\vx}^{(t+1)} \|_2^2 \\
        &\leq 2 \sum_{t=1}^T \|\vx^{(t)} - \tilde{\vx}^{(t)} \|_2^2 + 2 \sum_{t=1}^{T}  \|\vx^{(t)} - \tilde{\vx}^{(t+1)} \|_2^2.
    \end{align}
    Similarly, for Player $\cY$,
    \begin{align}
        \reg_\cY^{(T)} &\leq \left( 3 \etao_\cY \frac{L^2 \diam^2_\cY}{\delta_\cY} + \frac{L \diam_{\cY}^3}{\etao_\cY} \right) + \left( 3 \etao_\cY \frac{L^2 \diam_\cY}{\delta_\cY}  + \frac{L \diam_{\cY}^2}{\etao_\cY} \right) \sqrt{ \sum_{t=2}^T \| \vx^{(t)} - \vx^{(t-1)} \|_2^2 } \\
        &- \frac{\delta_\cY}{8 \etao_\cY} \sum_{t=2}^T \|\vy^{(t)} - \vy^{(t-1)} \|_2^2 -\frac{\delta_\cY}{4 \etao_\cY} \left( \sum_{t=1}^T \|\vy^{(t)} - \tilde{\vy}^{(t)} \|_2^2 + \sum_{t=1}^T  \|\vy^{(t)} - \tilde{\vy}^{(t+1)} \|_2^2 \right),\label{align:second-regy}
    \end{align}
    Using the fact that $\beta x - \alpha x^2 \leq \nicefrac{\beta^2}{4 \alpha}$ for $\alpha > 0$, we have
\begin{equation}
    \reg_\cX^{(T)} + \reg_\cY^{(T)} \leq \left( \beta_\cX(\etao_\cX) \diam_\cX + \beta_\cY(\etao_\cY) \diam_\cY + \frac{\beta^2_\cX(\etao_\cX)}{4 \alpha_\cY(\etao_\cY)} + \frac{\beta^2_\cY(\etao_\cY)}{4 \alpha_\cX(\etao_\cX)} \right),
\end{equation}
and the claim follows from~\Cref{prop:folklore}.
\end{proof}

\begin{remark}
    As discussed in the context of~\Cref{assumption:nomove}, assuming that $\delta_\cX > 0$ and $\delta_\cY > 0$ in~\Cref{cor:opt-AdOGD} is without any loss. If $\delta_\cX = \delta_\cY = 0$, then it follows that $(\vx^{(1)}, \vy^{(1)})$ is an exact equilibrium since $\vx^{(1)} \in \argmax_{\vx \in \cX} \langle \vx, \vu_\cX^{(1)} \rangle$ and $\vy^{(1)} \in \argmax_{\vy \in \cY} \langle \vy, \vu_\cY^{(1)} \rangle$, by definition of $\AdOGD$. Otherwise, let us assume that $\delta_\cX > 0$ and $\delta_\cY = 0$. Let $t$ be the first iteration such that $\vm^{(t)}_{\cY} \neq \vu^{(t)}_{\cY}$. For the duration of $\tau = 1, \dots, t-1$, Player $\cY$ incurs at most zero regret; this holds because each strategy of Player $\cY$ is a best response to the corresponding utility, by definition of $\AdOGD$ (since for all such $\tau$ we have $\vm_\cY^{(\tau)} = \vu_\cY^{(\tau)}$). Furthermore, for all $\tau = 1, \dots, t-1$, it holds that $\vu_{\cX}^{(\tau)}$ is constant since $\vy^{(\tau)}$ remains the same, again by definition of $\AdOGD$. Thus, by~\Cref{theorem:RVU-adOGD}, the regret of Player $\cX$ will also be bounded by a constant. From iteration $t$ onward, one reverts to our analysis in~\Cref{cor:opt-AdOGD}. The case where $\delta_\cX = 0$ and $\delta_\cY > 0$ is symmetric.
\end{remark}

We next turn to proving iterate convergence of $\AdOGD$. We follow the basic approach of~\citet{Anagnostides22:On}. Combining the analysis of~\Cref{cor:opt-AdOGD} together with~\Cref{fact:obvious}, it follows that the second-order path length of $\AdOGD$ is bounded.

\begin{corollary}[Bounded second-order path length for $\AdOGD$]
    \label{cor:pathlength}
    In the setting of~\Cref{cor:opt-AdOGD},
    \begin{equation}
        \left( \sum_{t=1}^T \|\vx^{(t)} - \tilde{\vx}^{(t)} \|_2^2 + \sum_{t=1}^T  \|\vx^{(t)} - \tilde{\vx}^{(t+1)} \|_2^2 \right) + \left( \sum_{t=1}^T \|\vy^{(t)} - \tilde{\vy}^{(t)} \|_2^2 + \sum_{t=1}^T  \|\vy^{(t)} - \tilde{\vy}^{(t+1)} \|_2^2 \right) = O_T(1). 
    \end{equation}
\end{corollary}
For the sake of exposition, we use the notation $O_T(\cdot)$ to suppress the dependence on parameters that do not depend on the time horizon $T$.

\begin{proof}[Proof of~\Cref{cor:pathlength}]
    Combining~\eqref{align:first-regx} and~\eqref{align:second-regy},
    \begin{equation}
        \label{eq:almost}
        \reg_\cX^{(T)} + \reg_\cY^{(T)} \leq \beta_\cX(\etao_\cX) \diam_\cX + \beta_\cY(\etao_\cY) \diam_\cY + \frac{\beta^2_\cX(\etao_\cX)}{4 \alpha_\cY(\etao_\cY)} + \frac{\beta^2_\cY(\etao_\cY)}{4 \alpha_\cX(\etao_\cX)} - 2 \alpha_\cX S_\cX^{(T)} - 2 \alpha_\cY S_\cY^{(T)},
    \end{equation}
    where we defined $S_\cX^{(T)} \defeq \sum_{t=1}^T \|\vx^{(t)} - \tilde{\vx}^{(t)} \|_2^2 + \sum_{t=1}^T  \|\vx^{(t)} - \tilde{\vx}^{(t+1)} \|_2^2 $ and $S_\cY^{(T)} \defeq \sum_{t=1}^T \|\vy^{(t)} - \tilde{\vy}^{(t)} \|_2^2 + \sum_{t=1}^T  \|\vy^{(t)} - \tilde{\vy}^{(t+1)} \|_2^2$. Combining~\eqref{eq:almost} with~\Cref{fact:obvious}, the claim follows.
\end{proof}

The first consequence of~\Cref{cor:pathlength} is that $\eta_\cX^{(T)} = \Theta_T(1)$ and $\eta_\cY^{(T)} = \Theta_T(1)$. Furthermore, after a sufficiently large number of iterations $T = O_\epsilon(1/\epsilon^2)$, there will exist an iterate $t \in [T]$ such that $\|\vx^{(t)} - \tilde{\vx}^{(t)} \|_2, \|\vx^{(t)} - \tilde{\vx}^{(t+1)} \|_2, \|\vy^{(t)} - \tilde{\vy}^{(t)} \|_2, \|\vy^{(t)} - \tilde{\vy}^{(t+1)} \|_2 \leq \epsilon$ (this actually holds for most iterates). By~\citet[Claim A.14]{Anagnostides22:On}, this implies that the strategy profile $(\vx^{(t)}, \vy^{(t)})$ has a duality gap of at most $O_\epsilon(\epsilon)$ since $\eta_\cX^{(T)} = \Theta_T(1)$ and $\eta_\cY^{(T)} = \Theta_T(1)$.

\begin{corollary}[Iterate convergence for $\AdOGD$]
    \label{cor:lastiterate}
    In the setting of~\Cref{cor:opt-AdOGD}, after $T$ iterations there is a strategy profile $(\vx^{(t)}, \vy^{(t)})$ with duality gap $O_T(T^{-1/2})$.
\end{corollary}

\section{Computing $\gamma$}\label{sec:gamma}
In this section, we give two different algorithms for computing the quantity $\gamma$ stipulated by \Cref{alg:prmp-mod}. For concreteness, our problem is the following: given a vector $\vv \in \R^n$ and a number $t > 0$, find the number $\gamma \in \R$ such that $\norm{[\vv - \gamma]_+}_2 = t$. First, note that the function $f(\gamma) := \norm{[\vv - \gamma]_+}_2$ is monotonically strictly decreasing in $\gamma$ for $\gamma < \max \vv$, and zero for $\gamma \ge \max_i v_i$; therefore, $f(\gamma) = t$ has a unique solution for every $t > 0$.

Both algorithms operate on the following premise: if $\vv^+ \in \R^k$ is the sub-vector of $\vv$ consisting of only elements larger than $\gamma$, then $\gamma$ satisfies $\norm{\vv^+ - \gamma}_2^2 = t^2$, and therefore
\begin{align}
    \gamma = \frac{1}{k} \ab(s - \sqrt{s^2 - k(s_2 - t^2)})\label{eq:gamma-quadratic}
\end{align}
where $s = \ip{\vec 1, \vv^+}$ and $s_2 = \ip{\vec 1, (\vv^+)^2}$, and $(\vv^+)^2$ denotes element-wise squaring.\footnote{If the quadratic has two roots, $\gamma$ must be the smaller of them, because the larger root is larger than $s/k$ and would hence violate the condition that $\vv^+ \ge \gamma$ element-wise.} Thus, it suffices to find the $k$ such that the $\gamma$ computed by solving \eqref{eq:gamma-quadratic} with the subvector $\vv^+$ consisting of the $k$ largest elements of $\vv$ satisfies 
\begin{align}\label{eq:gamma-constraints}
    \min \vv^+ \ge \gamma \ge \max \vv^-,
\end{align} where $\vv^- \in \R^{n-k}$ is the vector of remaining elements in $\vv$. 

The first algorithm is a sorting-based algorithm. If the elements of $\vv$ are sorted in descending order, then it suffices to loop over $\vv$, and for each possible subvector, compute \eqref{eq:gamma-quadratic} and check whether it is valid. This results in \Cref{alg:gamma1}.

The second algorithm is a selection-based algorithm: try setting $k=n/2$, and pivot to either the low or high subarrays based on which of the two inequalities in \eqref{eq:gamma-constraints} is violated. The resulting algorithm runs in linear time, assuming a linear-time selection algorithm such as that of \citet{Blum73:Time}.

\begin{algorithm}[h]
\caption{Computing $\gamma$ in $O(n \log n)$ time via sorting}\label{alg:gamma1}
$\vv \gets \vv$ with entries sorted in descending order \Comment{$O(n \log n)$ time}\\
$s \gets 0$\\
$s_2 \gets 0$\\
\For(\Comment{$1$-indexed}){$k = 1, \dots, n$}{
    $s \gets s + v_k$\\
    $s_2 \gets s_2 + v_k^2$\\
    $\gamma = \frac{1}{k} \ab(s - \sqrt{s^2 - k(s_2 - t^2)})$\\
    \lIf{$k=n$ or $\gamma \ge v_{k+1}$}{{\bf return} $\gamma$}
}
\end{algorithm}

\begin{algorithm}[h]
\caption{Computing $\gamma$ in linear time via selection}\label{alg:gamma2}
$s^+ \gets 0$ \\
$s^+_2 \gets 0$ \\
$k^+ \gets 0$\\
\Repeat{}{
    $n \gets{}$length of $\vv$\\
    $i = \floor{n/2}$\\
    $\vv \gets{}$partition($\vv, i$)\Comment{re-order $\vv$ so that $v_i$ is its $i$th smallest element. O(n) time}\\
    $\vv^-, \vv^+\gets \vv_{1:i}, \vv_{i+1:n}$\Comment{1-indexed, both bounds inclusive}\\
    $s \gets s^+ + \ip{\vec 1, \vv^+}$\\
    $s_2 \gets s_2^+ + \ip{\vec 1, (\vv^+)^2}$ 
    \Comment{element-wise squaring}\\
    $k \gets k^+ + (n-i)$\\
    $\ds \gamma \gets{} \frac{1}{k} \ab(s - \sqrt{s^2 - k(s_2 - t^2)})$\\
    \lIf(\Comment[2]{branch high\qq{}\qq{}\qq{}\qq{}\qq{}\qq{}}){$\gamma$ does not exist or $\gamma > v_i$}{%
        $\vv \gets \vv^+$
    }\lElseIf{$\gamma \ge \max \vv^-$}{{\bf return} $\gamma$}
    \lElse(\Comment[2]{branch low}){%
        $\vv, s^+, s^+_2, k^+ \gets \vv^-, s, s_2, k$
    }
}
\end{algorithm}

\section{Learning setups}
\label{sec:learningsetup}

\Cref{alg:learningsetups} gives the canonoical learning setups that we refer to throughout the paper---simultaneous iterates, alternating iterates, and extragradient---formulated for a general pair of no-regret learning algorithms $\cR_\cX$ and $\cR_\cY$.
\begin{algorithm}[h]
\let\oldnl\nl%
\newcommand{\nonl}{\renewcommand{\nl}{\let\nl\oldnl}}%
\caption{Canonical learning setups}\label{alg:learningsetups}
\nonl {\bf given:} \\
\nonl $\bullet$~~optimistic no-regret learning algorithms $\cR_\cX, \cR_\cY$  \\
\nonl \qq{} with functions {\sc NextStrategy} and {\sc ObserveUtility} \\
\nonl $\bullet$~~payoff matrix $\mA$\\
\nonl $\bullet$~~iteration limit $T$\\
\nonl $\bullet$~~initial predictions $\vu^{0}_\cX$, $\vu^{0}_\cY$ (\eg, $\vec 0$)\\
\Fn{\sc RunSimultaneousIterates}{
\For{$t = 1, \dots, T$}{
 $\vx^t \gets \cR_\cX$.{\sc NextStrategy}($\vu^{t-1}_\cX$)\\
 $\vy^t \gets \cR_\cY$.{\sc NextStrategy}($\vu^{t-1}_\cY$)\\ 
 $\vu^t_\cX \gets \mA \vy^t$\\
 $\vu^t_\cY \gets -\mA^\top \vx^t$\\
 $\cR_\cX$.{\sc ObserveUtility}($\vu^t_\cX$)\\
 $\cR_\cY$.{\sc ObserveUtility}($\vu^t_\cY$)\\
}
}
\Fn{\sc RunAlternatingIterates}{
\For{$t = 1, \dots, T$}{
 $\vx^t \gets \cR_\cX$.{\sc NextStrategy}($\vu^{t-1}_\cX$)\\
 $\vu^t_\cY \gets -\mA^\top \vx^t$\\
 $\cR_\cY$.{\sc ObserveUtility}($\vu^t_\cY$)\\
 $\vy^t \gets \cR_\cY$.{\sc NextStrategy}($\vu^{t}_\cY$)\\ 
 $\vu^t_\cX \gets \mA \vy^t$\\
 $\cR_\cX$.{\sc ObserveUtility}($\vu^t_\cX$)\\
}
}
\Fn{\sc RunExtraGradient}{
\For{$t = 1, \dots, T$}{ 
 $\tilde\vx^t \gets \cR_\cX$.{\sc NextStrategy}($\vec 0$)\\
 $\tilde\vy^t \gets \cR_\cY$.{\sc NextStrategy}($\vec 0$)\\ 
 $\vm^t_\cX \gets \mA \tilde\vy^t$\\
 $\vm^t_\cY \gets -\mA^\top \tilde\vx^t$\\
 $\vx^t \gets \cR_\cX$.{\sc NextStrategy}($\vm^{t}_\cX$)\\
 $\vy^t \gets \cR_\cY$.{\sc NextStrategy}($\vm^{t}_\cY$)\\ 
 $\vu^t_\cX \gets \mA \vy^t$\\
 $\vu^t_\cY \gets -\mA^\top \vx^t$\\
 $\cR_\cX$.{\sc ObserveUtility}($\vu^t_\cX$)\\
 $\cR_\cY$.{\sc ObserveUtility}($\vu^t_\cY$)\\
}
}
\end{algorithm}

\section{Further omitted proofs}

For completeness, we provide the standard proof showing that gradient descent has the one-step improvement property.

\begin{lemma}
    \label{lem:gd-onestepimprov}
    Let $\vx_i^{(t+1)} = \Pi_{\cX_i}(\vx_i^{(t)} + \eta_i^{(t)} \vu_i^{(t)})$. Then
    \[
       \max_{\vx_i' \in \cX_i} \langle \vu_i^{(t)}, \vx_i' - \vx_i^{(t)} \rangle \leq \frac{D_{\cX_i}}{\eta_i^{(t)}} \|\vx_i^{(t+1)} - \vx_i^{(t)} \|_2.
    \]
\end{lemma}

\begin{proof}
    By the first-order optimality condition,
    \[
        \left\langle \vu_i^{(t)} + \frac{1}{\eta_i^{(t)}} ( \vx_i^{(t+1)} - \vx_i^{(t)} ), \vx_i^{(t+1)} - \vx_i' \right\rangle \geq 0 \quad \forall \vx_i' \in \cX_i.
    \]
    As a result,
    \[
    \max_{\vx_i' \in \cX_i} \langle \vx_i' - \vx_i^{(t+1)}, \vu_i^{(t)} \rangle \leq \frac{1}{\eta_i^{(t)}} \langle \vx_i^{(t+1)} - \vx_i^{(t)}, \vx_i^{(t+1)} - \vx_i' \rangle \leq \frac{1}{\eta_i^{(t)}} \|\vx_i^{(t+1)} - \vx_i' \|_2 \|\vx_i^{(t+1)} - \vx_i^{(t)} \|_2,
    \]
    and the claim follows.
\end{proof}

We next move to $\IRPRMplus$; for a related one-step improvement property shown for $\RMplus$, we refer to~\citet{Anagnostides25:Convergence}.

\begin{lemma}\label{lem:onesteprmplus}
    $\IRPRMplus$ satisfies the one-step improvement property (\Cref{assumption:onestep}).
\end{lemma}

\begin{proof}
We will drop the superscripts and subscripts for notational cleanliness, and once again use $\lesssim$ to hide time-independent constants. Let
$\vr' = [\tilde\vr + \tilde\vg]_+$ and $\vx' = \vr'/\norm{\vr'}_1$. That is, $\vr'$ and $\vx'$ are the iterates that $\RMplus$ would take given utility $\tilde\vg$. Define the function $\pi : \R^d \to \R^d$ by $\pi(\vs) = \vs - \gamma \vec 1$, where $\gamma$ is such that $\norm{[\vs - \gamma \vec 1]_+}_2 = \norm{[\tilde \vr]_+}_2$. Then by definition we have $\vr = \pi(\vr')$, and $\tilde\vr = \pi(\tilde\vr)$, and moreover $\pi$ is Lipschitz. Therefore, we have
\begin{align}
    \norm{\vx - \tilde\vx} &\gtrsim \frac{\norm{\vr - \tilde\vr}}{\norm{\vr}} 
    \gtrsim \frac{\norm{\vr' - \tilde\vr}}{\norm{\vr}}
    \gtrsim \frac{\ip{\vr' - \tilde\vr, [\tilde\vg]_+}}{\norm{\vr} \cdot  \norm{[\tilde\vg]_+}}
    \gtrsim \frac{\norm{[\tilde\vg]_+}}{\norm{\vr}}
\end{align}
where we use, in order: Lipschitzness of the map $\vz \mapsto \vz / \norm{\vz}_1$ on the unit ball, Lipschitzness of $\pi$; Cauchy-Schwarz, and finally the fact that $\op{sgn}(\vr' - \tilde\vr) = \op{sgn}(\tilde\vg)$ and $\vr'[a^*] - \vr[a^*] = \tilde\vg[a^*]$, where $a^* = \argmax_a \tilde\vg[a]$. The proof is now complete by applying the standard analysis of regret matching, as per \eqref{eq:rm_regret_bound}, to upper-bound $\norm{\vr}$.
\end{proof}

\end{document}